\documentclass[12pt, a4paper, oneside]{amsart}
\usepackage{color}

\usepackage[utf8]{inputenc} 

\usepackage{amssymb}

\usepackage{bm}
\usepackage{bbm}

\usepackage{mathcomp}
\usepackage{dsfont}
\usepackage[toc,page]{appendix}
\usepackage{graphicx} 
\usepackage{subfigure}
\usepackage{enumerate}
\usepackage{amsfonts}

\usepackage{amscd}
\usepackage{enumitem}
\usepackage[mathscr]{eucal}
\usepackage{indentfirst}
\usepackage{graphicx}
\usepackage{graphics}
\usepackage{pict2e}
\usepackage{epic}
\usepackage{esint}
\usepackage[margin=2.9cm]{geometry}
\usepackage{epstopdf} 
\usepackage{verbatim}
\usepackage{cancel}
\usepackage{amssymb}
\usepackage{xcolor}
\usepackage{dsfont}
\usepackage{hyperref}
\usepackage{physics}
\usepackage{amsmath}

\usepackage{tikz}
\usepackage{float}
\usepackage{pgfplots}

\allowdisplaybreaks

\definecolor{green}{rgb}{0.0, 0.5, 0.5}
\definecolor{yellow}{rgb}{0.5, 0.5, 0}
\definecolor{lgray}{gray}{0.9}
\definecolor{llgray}{gray}{0.95}
\definecolor{lllgray}{gray}{0.975}

\usepackage{hyperref}

\setitemize{itemsep=5pt, topsep=7pt}
\numberwithin{equation}{section}

\newtheorem{definition}{Definition}[section]
\newtheorem{lemma}[definition]{Lemma}

\newtheorem{proposition}[definition]{Proposition}

\newtheorem{theorem}[definition]{Theorem}
\newtheorem{corollary}[definition]{Corollary}

\theoremstyle{remark}%
\newtheorem{remark}{Remark}%

\numberwithin{remark}{section}

\def\cN{\mathcal{N}}
\def\a0{{\rm a}_0}

\newcommand{\vertiii}[1]{{\left\vert\kern-0.25ex\left\vert\kern-0.25ex\left\vert #1 
		\right\vert\kern-0.25ex\right\vert\kern-0.25ex\right\vert}}

\newcommand{\inn}[2]{\left\langle#1,\,#2\right\rangle}

\DeclareMathOperator{\supp}{supp}

\newcommand{\Lap}{\Delta}
\newcommand{\di}{\partial}

\newcommand{\br}[1]{\bigl\langle#1\bigr\rangle}
\newcommand{\si}{\sigma}
\newcommand{\eps}{\epsilon}

\newcommand{\g}{\gamma}
\newcommand{\al}{\alpha}

\newcommand{\Zb}{\mathbb{Z}}
\newcommand{\Rb}{\mathbb{R}}
\newcommand{\N}{\mathcal{N}}
\newcommand{\cL}{\mathcal{L}}

\newcommand{\cU}{\mathcal{U}}

\newcommand{\W}{\mathcal{W}}

\newcommand{\om}{\omega}

\renewcommand{\l}{\lambda} 

\newcommand{\sbr}[1]{\left[#1\right]}
\newcommand{\Set}[1]{\left\{#1\right\}}

\newcommand{\md}[6]{\ensuremath{
		\ifinner
		\tfrac{\partial{^{#2}}#1}{\partial{#3^{#4}}\partial{#5^{#6}}}
		\else
		\tfrac{\partial{^{#2}}#1}{\partial{#3^{#4}}\partial{#5^{#6}}}
		\fi
}}
\newcommand{\del}[1]{\left(#1\right)}
\newcommand{\thmref}[1]{Theorem~\ref{#1}}

\newcommand{\secref}[1]{Section~\ref{#1}}
\newcommand{\lemref}[1]{Lemma~\ref{#1}}
\newcommand{\propref}[1]{Proposition~\ref{#1}}
\newcommand{\remref}[1]{Remark~\ref{#1}}

\DeclareMathOperator{\dG}{\mathrm{d}\Gamma}

\newcommand{\wndel}[1]{\inn{\psi}{{#1}\psi}}

\newcommand{\wtdel}[1]{\inn{\psi_t}{{#1}\psi_t}}

\newcommand{\xzdel}[1]{\inn{\xi_0}{{#1}\xi_0}}

\renewcommand{\cp}{\mathrm{c}}

\newcommand{\Fl}{\mathcal{F}_{\perp \varphi_t}^{\leq N}}

\newcommand{\cO}{\mathcal{O}}

\newcommand{\Int}{{\mathrm{int}}}

\title[ Bose-Einstein condensates in  disordered media]{On Bose-Einstein condensates in  disordered media}

\author{Marius Lemm}
\address{Marius Lemm, Department of Mathematics, University of T\"ubingen, 72076 T\"ubingen, Germany }
\email{marius.lemm@uni-tuebingen,de}

\author{Simone Rademacher}
\address{Simone Rademacher, Department of Mathematics, LMU Munich,  Theresienstrasse 39, 80333 Munich, Germany }
\email{simone.rademacher@math.lmu.de}

\author{Jingxuan Zhang}
\address{Yau Mathematical Sciences Center\\
	Tsinghua University\\
	Haidian District\\
	Beijing 100084, China }
\email{jingxuan@tsinghua.edu.cn}

\begin{document}
	\begin{abstract}
		We consider the quantum dynamics of interacting bosons in the mean-field regime when they are subjected to a disordered potential, which is either random or quasi-periodic. Starting from a spatially localized Bose-Einstein condensate, we prove that fluctuations around the condensate propagate with a small velocity due to the disorder. This provides an example of a disordered many-body system with provably slow transport behavior in any spatial dimension. The main technical novelty is an interaction picture analysis relative to the Anderson-localized one-body  dynamics. 
	\end{abstract}
	
	\maketitle 
	\date{\today}
	
	\section{Introduction}

	Disorder affects quantum  dynamics profoundly in ways that have no classical analog. The most famous example is that of Anderson localization \cite{anderson}, the complete absence of all  transport for a single quantum particle, which  occurs in any spatial dimension at sufficiently strong disorder. Rigorous proofs of this are landmarks works in mathematical physics: in the random case by Fr\"ohlich-Spencer \cite{frohlich1983absence} and by Aizenman-Molchanov \cite{aizenman1993localization}, and in the quasi-periodic case by Bourgain-Goldstein-Schlag \cite{BGS} and by Bourgain \cite{Bou}.  Both directions have grown into large subfields of mathematical physics  and we refer to the books \cite{bourgain2005green,aizenman2015random} and the review \cite{marx2017dynamics} for additional background. Pertinent directions for us are that a weaker form of dynamical localization  has also been proved for the nonlinear Schr\"odinger equation \cite{WZ2008,CSZ2020,CSW,CSWa} and that Anderson localization was extended to systems of multiple particles  \cite{chulaevsky2009multi,fauser2015multiparticle} with bounds that deteriorate as the particle number $N$ grows large. 
	
	In contrast to this large body of literature, the effect of disorder  on quantum \textit{many-body} systems is far less understood.
	By a many-body system, we mean a system with arbitrarily large particle number $N$, also known as positive-density systems. Understanding the quantum effects of disorder uniformly in $N$ remains a central problem for both mathematicians and condensed-matter physicists. 
	A particularly famous instance of this problem is to establish many-body localization (MBL) in  spatial dimension $d=1$ when the disorder is sufficiently strong compared to the interaction strength. There is a great deal of physics literature on MBL, but from the mathematical side, there are only a few rigorous results. In particular, for strongly disordered 1D spin chains, \cite{imbrie2016many} proved MBL under a level repulsion assumption and \cite{de2024absence} established  sub-diffusion. Other rigorous studies focused on 1D chains in spectral subspaces of low energy \cite{mastropietro2017localization,beaud2017low,elgart2018manifestations,elgart2024slow,elgart2024localization} or chains with heavy-tailed disorder \cite{gebert2022lieb,baldwin2023disordered,baldwin2025subballistic}. Other recent works considered stability of the MBL phase  under sparse perturbations \cite{toniolo2024stability,toniolo2025logarithmic} and MBL in hyperbolic lattices \cite{yin2024eigenstate}. 
	
	By contrast, in dimensions $d>1$, it is now commonly believed  that the  MBL phase is dynamically unstable \cite{de2017stability,luitz2017small}. Still, it is physically expected  that the coherence-destroying effects of disorder should lead to a substantial slowdown effect on the quantum many-body dynamics also for $d>1$. This leads to:\\

	\textit{Problem 1:} Show that disorder slows down many-body dynamics in any dimension.\\

	On the one hand, Problem 1 is related to studying MBL because it is also concerned with deriving bounds on transport from disorder. On the other hand, the focus is different: First, Problem 1 seeks to treat   any dimension. Second, a ``slow-down'' could be even a very modest decrease in propagation velocity.
	Methodologically, this problem is challenging and important because the existing analytical toolkit to prove any kind of dynamical slow-down is extremely limited. The reason for this is that all standard techniques for deriving ballistic upper bounds on quantum many-body transport, most notably Lieb-Robinson bounds \cite{lieb1972finite}, are completely insensitive to disorder. In other words, the  robustness of Lieb-Robinson bounds to the system details turns from a blessing to a curse when one is tasked with exploiting special features such as disorder. Indeed, exploiting disorder to improve Lieb-Robinson bounds has recently become an actively researched topics in 1D spin chains \cite{gebert2022lieb,baldwin2023disordered,toniolo2024stability,baldwin2025subballistic}. The recent work \cite{mcdonough2025non} then established the first general positive answer to Problem 1 by deriving a Lieb-Robinson bound with non-perturbatively small velocity in any spatial dimension. 
	
	In the present paper, we investigate Problem 1 for a paradigmatic model of quantum many-body dynamics --- interacting bosons in the mean-field regime \cite{hepp1974classical,ginibre1979classical,RS,lewin2015fluctuations}. We work on a lattice and the novel twist is that we place this system in a disordered environment modeled by a potential that is either random or quasi-periodic.\footnote{We mention in passing that discrete Bose gases in such disordered environments are physically realized in modern optical lattice experiments to study MBL \cite{choi2016exploring,sierant2018many}.} 
	We consider an initial state that is given by a spatially localized Bose-Einstein condensate and prove that its propagation is extremely slow: after time $t\gg 1$, the entire system is approximately contained within distances less than $\epsilon t$ from its initial location, where the speed $\epsilon\to 0$ with increasing disorder strength. 
	This shows strong slow-down due to disorder and thus addresses Problem 1 for mean-field bosons.\\ 
	
	We now continue the introduction with a slightly more precise, but still informal summary of the setup and main result. The mathematically precise versions are then given in Section \ref{sec:setup}. Afterwards, in Subsection \ref{sec:discussion} we discuss further how the result fits into the literature.

	\subsection{Summary of setup and main result }
	We consider $N$ bosons on $\mathbb Z^d$ in the mean-field regime. The bosonic nature is encoded in the choice of Hilbert space as the permutation-symmetric  elements\footnote{Here and in the following, all elements of $\ell^2$-spaces are take to be complex-valued functions.} of $\ell^2((\Zb^d)^N)$, i.e.,
	\[
	\psi_{N,t}(x_{\pi(1)},\ldots,x_{\pi(N)})=\psi_{N,t}(x_1,\ldots,x_N),\qquad \forall \pi \in S_N.
	\]
	The system's energy is encoded in the Hamiltonian, a linear self-adjoint operator of the form 
	\begin{equation}\label{eq:Hintro}
		H_N := \sum_{i=1}^N (-\Delta_{x_i}+\lambda v_{\omega}(x))+\frac{1}{N}\sum_{1\leq i,j\leq N} V(x_i-x_j).
	\end{equation}
	Here, $-\Delta_{x_i}$ denotes the discrete Laplacian acting in the $x_i$ variable,  $v_\omega(x_i)$ are multiplication operators that are, e.g., generated by a family of independent and identically distributed (i.i.d.) random variables, and $\lambda\in\mathbb R$ will be  a large parameter called the disorder strength. For the interaction, we consider for simplicity the ``on-site'' form  $V(x_i-x_j)=U\delta_{x_i,x_j}$,  with some $\mathcal{O}(1)$ constant $U\in\Rb$, but the proof generalizes to any $V:\mathbb Z^d\to\mathbb R$  of  compact support. The main challenge is that all bounds should be uniform in the particle number parameter $N$, which in particular means that the dimension of the Hilbert space is extremely large. The disorder strength $\lambda$ will be a large constant independent of $N$. 
	
	We study the quantum many-body dynamics associated to $H_N$. As a self-adjoint operator, $H_{N}$ generates the associated unitary operator $e^{-itH_{N}}$, which is the  solution operator to
	the $N$-body Schr\"odinger equation
	\begin{align}
		\label{eq:Schroe}
		i \partial_t \psi_{N,t} = H_N \; \psi_{N,t}.
	\end{align}
	The standard initial states for the dynamic \eqref{eq:Schroe} describe  uncorrelated Bose-Einstein condensates. Mathematically, this means one considers initial data given by a  product function of the form
	\begin{align}
		\label{eq:intro_initial}
		\psi_{N,0}(x_1,\ldots,x_N)=\prod_{i=1}^N\varphi_0(x_i) ,
	\end{align}
	for a fixed $\varphi_0\in \ell^2(\Zb^d)$. Due to the interactions, the product structure is not preserved under time evolution. A large subfield of mathematical physics  is concerned with showing that the many-body evolution approximately remains a product state; see \cite{hepp1974classical,ginibre1979classical,RS,lewin2015fluctuations} for results about mean-field dynamics and \cite{napiorkowski2023dynamics} for a review of work on the more singular Gross-Pitaevskii regime. The ``mean-field approximation'' is  mathematically captured through the $1$-particle reduced density matrix (or $1$-marginal in probabilistic terms) defined by
	\[
	\gamma_{\psi_{N,t}}(x,y)=\sum_{x_2,\ldots,x_N\in\mathbb Z^d}\psi_{N,t}(x,x_2,\ldots,x_N)\overline{\psi_{N,t}(y,x_2,\ldots,x_N)}.
	\]
	From the spectral theorem, one sees that $\gamma_{\psi_{N,t}}$ is comprised of the one-body wave functions in $\ell^2(\mathbb Z^d)$ that make up the many-body state $\psi_{N,t}$. Bose-Einstein condensation is mathematically associated with $\gamma_{\psi_{N,t}}$ having an eigenvalue proportional to $N$. For instance, at time $t=0$, the form of the initial state \eqref{eq:intro_initial} implies
	\[
	\gamma_{\psi_{N,0}}(x,y)=N\varphi_0(x)\overline{\varphi_0(y)}.
	\]
	Typical results, e.g.\ \cite{RS}, show that Bose-Einstein condensation is \textit{globally} preserved in the sense that for any bounded operator $O:\ell^2(\mathbb Z^d)\to \ell^2(\mathbb Z^d)$,
	\begin{equation}\label{eq:global}
		\big\vert \Tr \del{\big(\tfrac1N\gamma_{\psi_{N,t}} -   \vert \varphi_t \rangle \langle \varphi_t \vert\big)  O} \big\vert \leq C \|O \|_{\rm op}\frac{e^{Kt}}{N^{1/2}},
	\end{equation}
	where $\varphi_t\in \ell^2(\Zb^d)$ solves a suitable nonlinear Schr\"odinger equation with initial data $\varphi_0$ and this describes the dynamics of the highly populated condensate part. In situations where dispersive estimates hold, the $e^{Kt}$-factor can be removed \cite{DL}.
	
	We propose to consider \eqref{eq:global} and similar bounds to be of ``global'' nature, because they are independent of any spatial locality properties of either the condensate or the observable.

	\begin{figure}[H]
		\centering
		\begin{tikzpicture}[scale=.5]
			\draw (0,0) circle (1) node at (0,-.5) {$B_r$};
			
			\filldraw (0,0) circle (0.05);
			
			\draw[dashed] (0,0) circle (3.5*1.5);
			


			\draw[thick,<->] (0.7071,0.7071) to (2.47*1.5,2.47*1.5);
			
			\node at (0,2) {$B_R$};
			
			\node at (3.1,1.2) {$\rho=R-r$};
			
			
			
			
			\filldraw[gray!30] plot[shift={(-3.4,-1)},scale = 1.3, smooth, tension=1.2] coordinates {(-5.8,1.2) (-4.2,1.3) (-3,0) (-3.6,-1.8) (-6,-2) (-6.25,-0.2)};
			
			
			
			\node at (-9.5,- 1.5) {$\supp\varphi_0$ };

			\draw[dashed] (-9.5,-1.5) circle (3);
		\end{tikzpicture}
		\caption{The geometry in the main result (see \thmref{thm2}). We probe the system at time $t$ with a local observable $O$ acting on a small ball $B_r$ around the origin.    We assume that the initial condensate is supported outside of the larger ball $B_R$. In \cite{LRZa}, we prove that for \textit{any} external potential, the mean-field approximation is enhanced when $\rho> vt$ for an $\cO(1)$ constant $v$ independent of the potential.  Here we prove that for a disordered potential,   the mean-field approximation is enhanced when $\rho> \epsilon t$ for a small constant $\epsilon>0$ determined by the disorder strength $\l$.}\label{figBall1}
		
	\end{figure}
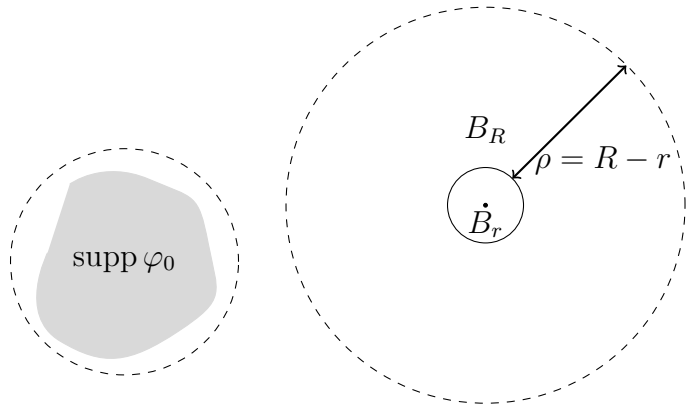
	
	In our recent work \cite{LRZa}, we developed a new perspective on mean-field dynamics by exploring its \textit{spatial locality structure}. For this, we considered a geometric setup as depicted in Figure \ref{figBall1} in which we probe the system at time $t$ with a \textit{local} bounded observable $O$ that only acts on a small ball $B_r$, centered at the origin for simplicity. We assume that the initial state \eqref{eq:intro_initial} characterized by $\varphi_0\in \ell^2(\mathbb Z^d)$ is supported at distance $\rho>0$ from the region where $O$ acts. In \cite{LRZa}, we considered the Hamiltonian \eqref{eq:Hintro} \textit{without} disordered environment ($v_\omega=0$) in the  geometric setup shown in Figure \ref{figBall1}  and we derived a local  enhancement of the mean-field approximation:
	\begin{equation}\label{eq:intro_old}
		\big\vert \Tr \big(\big(\tfrac{1}{N}\gamma_{\psi_{N,t}} -  \vert \varphi_t \rangle \langle \varphi_t \vert\big)  O\big) \big\vert \leq 
		\frac{C_n\norm{O}_{\mathrm{op}}	}{N}
		\frac{ 1  }{\rho^n},   \qquad \textnormal{for } t\le \frac\rho v.
	\end{equation}
	where $n\geq 1$ can be taken arbitrarily large.  
	We see that the right-hand side decays super-polynomially in the distance parameter $\rho$ as long as $\rho>   vt$. Thus \eqref{eq:intro_old} can be interpreted as a \textit{ballistic} propagation bound on the quantum many-body dynamics: Within time $t$, fluctuations around the mean-field behavior (i.e., the Hartree states) propagate at most a distance $\rho\sim vt$ up to super-polynomially small errors. Here, $v>0$ is an $\mathcal O(1)$ constant  that plays the role of a velocity bound. The result of \cite{LRZa} thus establishes that the physical principle of locality that information spreads at most with a finite system-dependent speed also applies to mean-field bosons and to fluctuations around the condensate. This can also be viewed as the existence of an effective light cone with slope $v>0$; see Figure \ref{fig:lccompare}.
	To obtain \eqref{eq:intro_old}, \cite{LRZa} developed a new adaptation of the ASTLO method (which was originally developed to control the dynamics of Bose-Hubbard Hamiltonians without mean-field scaling and $\mathcal O(1)$ local particle density \cite{FLS,FLSa,LRSZ,LRZ,lemm2025quantum}) that could handle the particle non-conserving generator of the fluctuation dynamics.
	
	In the present paper, we consider the same geometric setup as depicted in Figure \ref{figBall1}. We add disorder and prove a qualitatively significantly stronger propagation bound. Indeed, we prove that for any $\epsilon>0$, one can take the disorder strength sufficiently large so that
	\begin{equation}
		\big\vert \Tr \big(\big(\tfrac{1}{N}\gamma_{\psi_{N,t}} -  \vert \varphi_t \rangle \langle \varphi_t \vert\big)  O\big) \big\vert\leq  {C\norm{O}_{\mathrm{op}}} \left(\frac{e^{K(t-\rho/\eps)}}{N}+\frac{e^{Kt}}{N^2}\right),\qquad \textnormal{for } t \leq  \frac{\rho}{ \epsilon}.
	\end{equation}
	To understand the result and compare it to prior ones, we consider large $N$ and we take $\rho$ and $t$ to be large numbers, but independent of $N$. The second term is $\mathcal O(N^{-2})$ and can therefore be ignored for large $N$ for comparison purposes. 
	The first term is of the form $N^{-1}$ times the crucial new factor $e^{K(t-\rho/\eps)}$, which is exponentially small for  $\rho >\epsilon t$. The condition $\rho >\epsilon t$ corresponds to the outside of a very slim light cone; see Figure \ref{fig:lccompare} for the comparison to the prior light cone established in \cite{LRZa}.

	\begin{figure}
		\resizebox{10cm}{!}{
			\begin{tikzpicture}
				\begin{axis}[
					axis lines = middle,
					xlabel={$\abs{x}$},
					xmin=0, xmax=5, ymin=0, ymax=5,
					xtick=\empty, ytick=\empty,
					legend pos=north west
					]
					
					\addplot[fill=gray!30, opacity=0.5]
					coordinates {
						(0,0) (2,1) (4,2) (5,2.5) (5,5) (0,5) (0,0)
					} --cycle;
					
					\addplot[fill=gray!90, opacity=0.5]
					coordinates {
						(0,0)
						(0.5,5)
						(0,5) (0,0)
					} --cycle;
					
					\addplot[thick, domain=0:0.5] {10*x};
					
					\addplot[thick, domain=0:5] {0.5*x};

				\end{axis}
				\node at (4.8,4) {$\mathcal{O}(1)$ slope};
				\node[] at (4.3,3.6) {(prior work \cite{LRZa})};
				\node[left] at (-.5,4.2) {slope $\eps$};
				\node[left] at (-.5,3.8) {(Thm.~\ref{thm1})};
				\draw [->] (-.5,4) to (.3,4);
				\node[left] at (0,5.5) {time $t$};
			\end{tikzpicture}
			
		}
		\caption{Our main result establishes the effective space time light cone $\abs{x}\le \eps t$ with $\eps>0$ arbitrarily small for sufficiently strong disorder (darker shaded region), outside of which  propagation of the condensate and quantum fluctuations around it are exponentially suppressed. }
		\label{fig:lccompare}
	\end{figure}
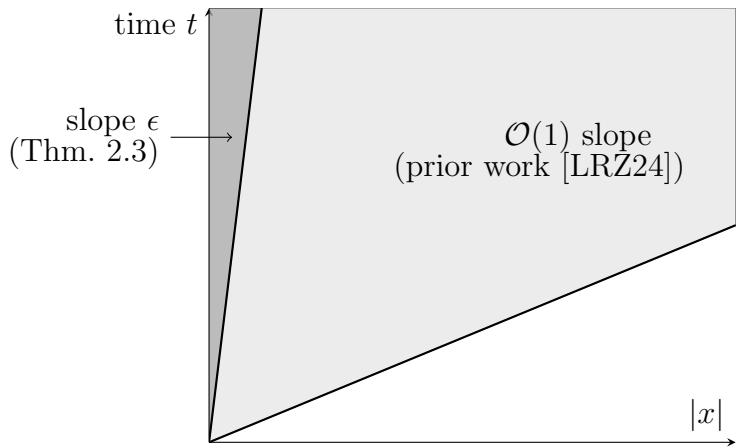

	We note that the constant $K$ in the second term is exactly the same constant as in the global bound by Rodnianski-Schlein \cite{RS}. However, for us the exponential is multiplied with $N^{-2}$ instead of $N^{-1}$, so the global term is rather harmless now. In contrast to \cite{DL,LRZa}, we need to use the global bound on the fluctuation number from \cite{RS} within an error term because dispersive estimate on the nonlinear evolution are  no longer available to us 
	because the  disorder prevents the nonlinear Schr\"odinger equation from dispersing.

	\subsection{Organization of the paper}
	In Section \ref{sec:setup}, we formulate the precise mathematical setup and state the main theorems. 
	
	In Section \ref{sec:flucdyn}, we recall the standard formalism of the so-called fluctuation dynamics \cite{RS, LNSS15}. 
	
	In Section \ref{sec:proofthm1}, we prove the crucial bound Theorem \ref{thm1} on the local fluctuation number. This section is the heart of the paper and contains the key methodological novelty, the development of a suitable interaction picture relative to the Anderson-localized one-body dynamics.
	
	In Section \ref{secPfThm2}, we briefly review how to conclude the main result on the mean-field approximation, Theorem \ref{thm2}, from Theorem \ref{thm1} by standard methods.

	\section{Setup and main results}\label{sec:setup}
	
	We study the quantum many-body dynamics of bosonic systems on lattices. 
	We formulate the problem as usual in terms of the algebraic operator formalism of creation and annihilation operators, i.e., we consider the permutation-symmetric Fock space 
	\begin{align}
		\label{def:Fockspace}
		\mathcal{F}(\ell^2(\mathbb Z^d))
		=\mathbb C\oplus \bigoplus_{N\geq 1} (\ell^2(\mathbb Z^d))^{\otimes_s N},
	\end{align}
	with $d\geq 1$ representing the spatial dimension and $\otimes_s N$ denoting the $N$-fold symmetrized tensor product. On the Fock space, we consider self-adjoint many-body Hamiltonian operator given by 
	\begin{align}\label{HNdef}
		H_N=	H_{N } ^{\l,\om,U}=  \sum_{\substack{x,y \in \Zb^d  }  } (h^{\l,\om}_{x,y}) \;  a_x^*a_y+ {\frac{U}{ 2N}}\sum_{\substack{x \in \Zb^d  }} \; a_x^*a_x^*a_xa_x \; . 
	\end{align}
	Here, the operator $h^{\l,\om}=-\Lap+\l v_\om$ is a disordered $1$-body Hamiltonian acting on $\ell^2(\Zb^d)$, where   $-\Lap$ is the discrete Laplacian, and  $v_\om\in\ell^\infty(\Zb^d)$ is a disordered potential subject to Assumptions \ref{C1} and \ref{C2} below. The disorder strength $\l>0$  will be a large constant for us. Explicitly, denoting by $x\sim y$ the nearest neighbors on $\mathbb Z^d$, i.e.  $\sum_{i=1}^d\abs{x_i-y_i}=1$, we have
	\begin{align}
		\label{hDef}
		h^{\l,\om}_{x,y} = \begin{cases}
			-1 & \text{if} \quad x\sim y,\\
			\l	v_\om(x) & \text{if}\quad  x=y,\\
			0& \text{otherwise}.
		\end{cases}
	\end{align}
	In the second term in \eqref{HNdef},   $U\in\Rb$ is an order $1$ constant measuring the interaction strength, and   $N$ is a large positive constant representing the particle number.  The prefactor $\frac{1}{N}$ in front corresponds to the paradigmatic \textit{mean-field scaling}.

	We consider the Schr\"odinger evolution \eqref{eq:Schroe} of an initially purely factorized state 
	\begin{align}
		\label{def:initial}
		\psi_{N,0} = \varphi_0^{\otimes N} \quad \text{with}\quad   \varphi_0 \in \ell^2( \Zb^d),\,\norm{\varphi_0}_{\ell^2}=1 .  
	\end{align}
	Due to interaction, the Schr\"odinger evolution 
	does not preserve the factorization for $t>0$. However, a macroscopic fraction of the particles is well-described by a $1$-particle nonlinear effective evolution given by the Hartree equation\begin{align}
		\label{NLS}
		i \partial_t \varphi_t = (h^{\l,\om}+  { U\vert \varphi_t \vert^2})\varphi_t  .
	\end{align}
	The main goal in this paper is to show that the propagation of quantum fluctuation slows down in the presence of sufficient disorder. In particular, the propagation speed for quantum fluctuation around Hartree states can be made arbitrarily slow. \\

	\paragraph{\textbf{Notation.}}	For $t\geq 0$, we denote by $\varphi_t$  the solution to the Hartree equation \eqref{NLS} with initial data $\varphi_0$,    $\psi_{N,t}$ the solution to the Schr\"odinger equation \eqref{eq:Schroe} with initial data $\psi_{N,0}$, and $\gamma_{\psi_{N,t}}$ the associated reduced $1$-particle  density matrix satisfying 
	$$\Tr  (\gamma_{\psi_{N,t}}   O )
	=  \langle \psi_{N,t}, \dG (O) \psi_{N,t} \rangle,$$
	for any $1$-particle operator $O$, where $\dG$ denotes the second quantization map.
	(In particular, $\Tr \g_{\psi_{N,t}} \equiv N$.)
	We abbreviate ${\ell^p}\equiv{\ell^p(\Zb^d)}$ for $1\le p\le \infty$,  $B_r:=\Set{x\in\Zb^d:\abs{x}\le r}$ for $r>0$, $B_r^\cp:=\Set{x\in\Zb^d:\abs{x}> r}$, and  $\norm{\cdot}_{\mathrm{op}}$  for the operator norm  on $\ell^2$. 
	
	\subsection{Assumptions}
	Our main results hold under the following  localization conditions of the associated $1$-body and nonlinear dynamics. We show in \secref{sec13} how these conditions  {are 
		verified for concrete disordered potentials in arbitrary dimension $d\ge1$, using existent results in the literature.}

	\begin{enumerate}[label=$\mathrm{(C1)}$]
		\item  \label{C1} We assume  the $1$-body Hamiltonian  $h=h^{\l,\om}$ from \eqref{hDef} satisfies  the semi-uniform dynamical localization (SUDL) condition as follows: 
		There exists $\l_0>0$ such that for all $\l\ge \l_0$, there exists $\g=\g(\l)\ge1$ such that for any $0<b\le1$, there exists $C_b>0$ such 
		\begin{align}
			\label{SUDL}
			\sup_{t\ge0}\abs{\br{e^{-ith} \delta_x, \delta_y}}\le C_b e^{b\abs{x}-\g\abs{x-y}}.
		\end{align}
		Furthermore, the decay rate satisfies $\lim_{\l\to\infty}\g(\l)=\infty$.
	\end{enumerate}	
	\begin{enumerate}[label=$\mathrm{(C2)}$]
		\item  \label{C2} We assume the NLS  \eqref{NLS} satisfies the  long-time Anderson localization as follows: Let $R\ge 2r>0$ and $\varphi_0(x)=0$ for $\abs{x}\le R$. Then, for any $\eps,\,M>0$, there exists $\l_1=\l_1(\eps, M,R-r)>0$ such that for all $\l\ge \l_1$, 
		\begin{align}
			\label{}
			\label{NLSDL}
			\sup_{0\le t\le (R-r)/\eps}	\sum_{\abs{x}\le r} \abs{\varphi_t(x)}^2 < e^{-M (R-r)/\eps}.
		\end{align}

	
\end{enumerate}

Concerning Condition  (C1), we mention that even in $d=1$, SUDL is  known to be delicate;  see Theorem 1.8 in \cite{JLM} and the classical results \cite{dRJLS,JL}. The growth in $x$ in SUDL \eqref{SUDL} is unavoidable if one wants to avoid averages and instead aims for statements that hold with high probability, as we do here. This is because any region has a small chance of being a rare region in which the potential behaves uncharacteristically.
Fortunately, since  we are only interested in $x$ from a given ball region,  the prefactor in front of \eqref{SUDL} can be made uniform in terms of the ball radius. I.e., the chance that the specific region where we look behaves uncharacteristically is small.

In \ref{C2}, the constant $\eps \ll 1$ is the effective velocity of the nonlinear evolution and it can be made arbitrarily small. This is an important ingredient for us, as a basis for controling the many-body fluctuations.  Finally, 
$M>0$ is an $\cO(1)$ constant.

\subsection{Main results}

Our first main result is a locally enhanced  mean-field approximation  for local observables that confirms the existence of dynamical slow-down due to disorder.

\begin{theorem}[Main result --- slow spreading of mean-field error] \label{thm2} Let $r>0$,  $R\ge 2r$, and $\rho=R-r.$
	Assume  that Conditions \ref{C1} and \ref{C2} hold and assume further that
	\begin{itemize}
		\item[(i)]  $\psi_{N,0}= \varphi_0^{\otimes N}$ is purely factorized;
		\item[(ii)] $\varphi_0$ has compact support, is normalized to $\norm{\varphi_0}_{\ell^2}=1$, and satisfies
		\begin{align}
			\varphi_0 (x) = 0  \quad \text{for} \quad \abs{x}\le R\; . 
		\end{align}
	\end{itemize}
	
	Then, given any $\eps>0$, there exist  $\l_*=\l_*(\eps, \rho,d)>0$,  $C = C(r,d)>0$, and $K=K(U)>0$  such that for all $\l\ge\l_*$,  all $0\le t \le \rho/\eps$, and any bounded local operator $O$ acting on $\ell^2$ with kernel satisfying 
	\begin{align}
		\label{ass:O}
		O(x,y) = \mathds{1}_{x \in B_r} O(x,y) \mathds{1}_{y \in B_r}  ,
	\end{align}
	it holds that
	\begin{align}\label{114}
		\big\vert \Tr \del{\big(\tfrac1N\gamma_{\psi_{N,t}} -   \vert \varphi_t \rangle \langle \varphi_t \vert\big)  O} \big\vert \leq  {C \|O \|_{\rm op}} \del{ \frac{e^{K(t-\rho/\eps)}}{N}+\frac{e^{Kt}}{N^2}}.
	\end{align}
\end{theorem}
This theorem is proved in \secref{secPfThm2}.   Typically, the larger the disorder strength $\l$ is, the smaller  the effective velocity $\eps$ can be made. For certain concrete models, the relationship between $\eps$ and $\l$ can be made explicit. For example, for the quasi-periodic models considered in \secref{sec231}, $\eps$ can be chosen to be   $\cO(\frac{1}{\log\l})$. See \remref{remSize} for details.

A corollary of \thmref{thm2} and assumption \ref{C2} is an improved estimate on  $\Tr(\tfrac{1}{N}\g_{\psi_{N,t}}O)$, which is a priori of order $\cO(1)$.

\begin{corollary}[Many-body propagation bound]
	Let the assumptions of  \thmref{thm2} hold. Then,  given any $\eps,\,M>0$, there exists  $\l_*=\l_*(\eps,M, \rho,d)>0$   such that for all $\l\ge\l_*$,  all $0\le t \le \rho/\eps$, and any   local $1$-body observable $O$ satisfying \eqref{ass:O}, there holds
	\begin{equation}\label{2111}
		\big\vert \Tr \big(\tfrac{1}{N}\gamma_{\psi_{N,t}} O\big) \big\vert 
		\leq   {C \|O \|_{\rm op}} \del{e^{-M\rho/\eps}+ \frac{e^{K(t-\rho/\eps)}}{N}+\frac{e^{Kt}}{N^2}}.
	\end{equation}
	Here $C$ and $K$ are as in \eqref{114}.
\end{corollary}
\begin{proof}
	We commence with
	\begin{align}
		\label{}
		\abs{\Tr(\tfrac{1}{N}\g_{\psi_{N,t}}O)} \le& \big\vert \Tr \del{\big(\tfrac1N\gamma_{\psi_{N,t}} -   \vert \varphi_t \rangle \langle \varphi_t \vert\big)  O} \big\vert+\big\vert \Tr \del{  \vert \varphi_t \rangle \langle \varphi_t \vert   O}\big\vert.
	\end{align}
	The first term in the r.h.s.~is bounded as in \eqref{114}.  By assumption \eqref{NLSDL}, for $t\le \rho/\eps$ and $O$ localized in $B_r$, the second term is bounded as $\big\vert\Tr \del{  \vert \varphi_t \rangle \langle \varphi_t \vert   O}\big\vert =\abs{\inn{\varphi_t}{O\varphi_t}}\le \norm{O}_{\mathrm{op}}\sum_{\abs{x}\le r} \abs{\varphi_t(x)}^2 < \norm{O}_{\mathrm{op}} e^{-M (R-r)/\eps}$. Thus the claim follows. 
\end{proof}

By standard techniques, the proof of	\thmref{thm2} reduces to a propagation bound on fluctuations, which we now state  as a separate theorem, in part because it holds for a more general class of initial states.
Define the number of fluctuations orthogonal to the condensate,
\begin{align}\label{N+def}
	\mathcal{N}^+ (t) = \sum_{i=1}^N q_t^{(i)},
\end{align}
where $q_t^{(i)}$ is the operator on  $\ell_s^2( \Zb^{dN}):=[\ell^2(\Zb^d)]^{\otimes _s N}$, acting as $q_t^{(i)} = 1 - \vert \varphi_t \rangle \langle \varphi_t \vert$ on the $i$-th particle and as identity on all other particles.

Our second theorem says that the \textit{local} number of excitations in a ball $B_r$ of radius $r>0$ around the origin,
\begin{align}
	\mathcal{N}_{B_r}^+ (t) = \sum_{i=1}^N  \big( q_t \mathds{1}_{B_r} q_t \big)^{(i)},\label{Nr} 
\end{align}
satisfies a propagation bound with arbitrarily small propagation speed $\epsilon$. In contrast to \thmref{thm2}, for the next result, we do not require the initial state to be a pure product state as in \eqref{def:initial}.

\begin{theorem}[Slow propagation of fluctuations] \label{thm1}

Let $r>0$,  $R\ge 2r$, and $\rho=R-r.$
	Assume that Conditions \ref{C1} and \ref{C2} hold and assume further that
	\begin{itemize}
		\item[(i)] $\psi_{N,0}$ has no fluctuations in $B_{R}$, i.e.,
		\begin{align}
			\label{psi0Cond}
			\big\langle \psi_{N,0}, \; \mathcal{N}_{B_{R}}^+ (0) \;  \psi_{N,0} 
				\big\rangle  
				=0.
			\end{align}
			\item[(ii)] $\varphi_0$ has compact support, is normalized to  $\norm{\varphi_0}_{\ell^2}=1$, and satisfies 
			\begin{align}\label{phi0cond}
				\varphi_0 (x) = 0  \quad \text{for} \quad \abs{x}\le R\; . 
			\end{align}
		\end{itemize}

		Then, given any $\eps>0$, there exist  $\l_*=\l_*(\eps, \rho,d)>0$,  $C = C(r,d)>0$, and $K=K(U)>0$  such that for all $\l\ge\l_*$ (see \eqref{hDef}) and all $0\le t \le \rho/\eps$,
		\begin{align}
			\label{slowPE}
			\br{\psi_t, \cN^+_{B_r}(t)\psi_t} \le 	 & C\sbr{ \del{e^{K(t-\rho/\eps) }+\frac{e^{Kt}}{N}}\br{(\cN_{B_R^\cp}+1)^2 }_0+{e^{-\g\rho}\br{\cN_{B_R^\cp}}_0}}.
		\end{align}
		Here $\g=\g(\l)>0$ is as in \ref{C2}.
	\end{theorem}
	This theorem is proved in \secref{secPfThm1}. 	We interpret \eqref{slowPE} as saying that for sufficiently large $N$ and $\rho$,  the local fluctuation around the Hartree states in the region $\abs{x}\le r$ near the origin is suppressed for times $t\le \rho/\eps$. This effectively says that the fluctuation propagates at arbitrarily slow speed for sufficiently strong disorder strength.

	\begin{remark}
		\begin{itemize}
			\item[(i)] If the initial state is purely factorized, i.e., if \eqref{psi0Cond} holds for all $R$, then \eqref{slowPE} simplifies to
			\begin{align}
				\label{slowPEsim}
				\br{\psi_t, \cN^+_{B_r}(t)\psi_t} \le 	 & C{ \del{e^{K(t-\rho/\eps) }+\frac{e^{Kt}}{N}}  }.
			\end{align}
			\item[(ii)]
			The condition $R\ge2r$ in Thms.~\ref{thm2} and \ref{thm1} can be replaced by $R\ge \al r$ for any $\al\in(0,2]$, in which case the threshold $\l_*$ will depend  in addition on $\al$. 
		\end{itemize}
		
	\end{remark}

	\subsection{Examples}\label{sec13}

	In this subsection, we present  concrete examples for which Conditions \ref{C1} and \ref{C2} have been proved. We consider two types of models ---  random potentials and quasi-periodic potentials, both in arbitrary dimension $d\ge1$.

	\subsubsection{Quasi-periodic potential}
	\label{sec231}
	
	Consider quasi-periodic Sch\"odinger operators of the form
	$h=-\Lap+\l v_{\theta,\al}(x)$ on $\Zb^d$, where
	\begin{align}
		\label{vCan}	v_{\theta,\al}(x)= \cos(2\pi (\theta+x\cdot \al)),\quad (\theta,\al)\in [0,1]\times[0,1]^d.
	\end{align}
	For $\al$ satisfying the Diophantine condition,  the following proposition is a direct consequence of the exponential dynamical localization  for \eqref{vCan} proved in \cite{GYZa} (see also \cite{CSZc,CSZd}), together with the long-time Anderson localization for \eqref{NLS} with quasi-periodic potential proved in \cite{CSW} (see also \cite{GYZ}). 	See \secref{secPfEx1} for details.
	
	\begin{proposition}\label{Ex1}
		For any $\delta_*>0$, there exists a subset of $(\theta,\al)$ of measure at least $1-\delta_*$ such that Conditions \ref{C1} and \ref{C2} hold.
	\end{proposition} 


	\subsubsection{Random potential}
	Next, we consider random Sch\"odinger operators of the form
	$h=-\Lap+ v_\omega(x)$ on $\Zb^d$, where  $v_\omega(x)$ acts as multiplication by a sequence of independent identically distributed (i.i.d.) random variables $\omega\in [0,1]^{\Zb^d}$, with uniform distribution. 
	\begin{proposition}\label{Ex2}
		For any $\delta_*>0$, there exists a subset of $\om$ of measure at least $1-\delta_*$ such that Conditions \ref{C1} and \ref{C2} hold.
	\end{proposition}
	
	For the proof, we refer to Aizenman's classical result for random Schr\"odinger operators, and a recent long-time Anderson localization result for random NLS \cite{CSWa}. See \secref{secPfEx1} for details.

	\subsection{Discussion}\label{sec:discussion}
	In this work, we established that Bose-Einstein condensates in the mean-field regime propagate arbitrarily slowly if they are subjected to sufficiently strong disorder.
	
	The general theme of our work bears similarity to the recent work \cite{mcdonough2025non} because both works tackle Problem 1 formulated in the introduction. However, we would like to emphasize two important differences between our work and \cite{mcdonough2025non}. First, the physical systems are completely different: The localized Bose-Einstein condensate we consider has all of its degree of freedom localized together in space, while the quantum spin systems studied in \cite{mcdonough2025non} have their the degrees of freedom distributed in a spatially homogeneous way across the whole system.  Second, the analytical methods employed are also completely different: Here we develop an interaction picture analysis relative to the Anderson-localized one-body dynamics (more on that later), while \cite{mcdonough2025non} develops a combination of Lieb-Robinson bounds and local Schrieffer-Wolf transformations based on a level repulsion assumption on the non-interacting Hamiltonian.

	Indeed, in contrast to \cite{mcdonough2025non}, in the present work we heavily rely on prior developments in two traditional areas of mathematical physics: Firstly, localization results on the one-body and nonlinear dynamics which we formulated as Conditions \ref{C1} and \ref{C2} associated with the works \cite{BGS,Bou,WZ2008,CSZ2020,CSW} and, secondly, the perturbative analysis of mean-field Bose dynamics \cite{hepp1974classical,ginibre1979classical,RS,lewin2015fluctuations}. Our work thus connects these two long-standing areas in mathematical physics to study the effects of disorder on many-body dynamics.

	This work starts from the local perspective on mean-field dynamics that we 
	developed in  \cite{LRZa},  which allowed us to prove that the Bose gas spreads at most ballistically in general settings (i.e., without disorder). To achieve this, in \cite{LRZa} we adapted the ASTLO method to control the generator of the fluctuation dynamics, thus treating a particle non-conserving generator which created different technical challenges compared to \cite{FLS,FLSa,LRSZ,LRZ,lemm2025quantum}. Ballistic propagation bounds also played a crucial role in the study of mean-field dynamics for high-density Fermi gases \cite{fresta2023effective,fresta2024effective}.  
	
	Here, we make a new methodological step forward beyond \cite{LRZa} by identifying a way to incorporate and exploit  disorder, a  notoriously difficult task in interacting systems. Our main technical contribution is a suitable interaction picture analysis of the fluctuation dynamics that allows us to incorporate  localization results \ref{C1} and \ref{C2} in the many-body setting. We emphasize that using the strength of Anderson localization properties of Anderson type in a perturbative argument is of course a natural idea, but a naive approach is usually obstructed by the fact that the interactions are simply not a sufficiently small perturbation. Our point is that such a perturbation argument can indeed be implemented through a suitable locality scheme (loosely related to the ASTLO method) when one is asking a local question about fluctuations in the mean-field regime. 
	
	Finally, as an open problem, we mention extending the results presented here to the more singular Gross-Pitaevskii regime; see \cite{napiorkowski2023dynamics} for  a recent review.

	\section{Fluctuations around the Hartree evolution} 
	\label{sec:flucdyn}

	This section is devoted to an overview of the properties of the dynamics governing the quantum fluctuations around the Hartree states. For details and the proof, we refer to \cite{LRZa}.

	The fluctuation dynamics is defined on the truncated Fock space 
	\begin{align}\label{FlDef}
		\mathcal{F}_{\perp \varphi_t}^{\leq N},   :=   \bigoplus_{j=0}^N \ell^2_{\perp \varphi_t}( \mathbb{Z}^d)^{\otimes_s j}  
	\end{align}
	and formulated w.r.t.~to the excitation map $\mathcal{U}_{N,t} : \ell^2_s ( \mathbb{Z}^{dN} ) \mapsto \mathcal{F}_{\perp \varphi_t}^{\leq N}$, which acts on any $\psi_N \in \ell^2_s ( \mathbb{Z}^{dN} ) $ with the unique decomposition 
	\begin{align}
		\psi_N = \sum_{j=0}^N \varphi_t^{\otimes (N-j)} \otimes_s \xi_t^{(j)}, \quad \text{where} \quad \xi_t^{(j)} \in \ell^2_{\perp \varphi_t} ( \mathbb{Z}^d)^{\otimes_s j}, 
	\end{align}
	on the corresponding excitation vector, i.e.  {$\mathcal{U}_{N,t} \psi_N = \lbrace \xi_t^{(0)}, \dots, \xi_{t}^{(N)} \rbrace $.}   Based on this excitation map, the fluctuation dynamics is defined by 
	\begin{align}
		\mathcal{W}_N (t;s) = \mathcal{U}_{N,t} e^{-iH_N (t-s)}\mathcal{U}_{N,s}^* \label{def:W} \; .  
	\end{align}
	A straight forward computation, using the properties of the unitary $\mathcal{U}_{N,t}$ (see for example \cite{LNSS15}), shows that the fluctuation dynamics satisfies 
	\begin{align}\label{Weq}
		i \partial_t \mathcal{W}_N (t;s) = \mathcal{L}_N (t) \mathcal{W}_N (t;s) 
	\end{align}
	and the generator $\mathcal{L}_N (t)$ is given by 
	\begin{align}
		\label{eq:L-sum}
		\mathcal{L}_N (t) := \mathbb{H} + \sum_{j=1}^N \mathcal{R}_{N,t}^{(j)} \; . 
	\end{align}
	The leading order term $\mathbb{H}$ of the generator is quadratic in modified creation and annihilation operators, which are defined on $\mathcal{F}_{\perp \varphi_t}^{\leq N}$ for any $f \in \ell^2( \mathbb{Z}^d)$ by 
	\begin{align}
		\label{def:b}
		b^*(f) = a^*(f) \sqrt{1- \mathcal{N}/N}, \quad \text{resp.} \quad b(f) = \sqrt{1- \mathcal{N}/N} a(f) \; . 
	\end{align}
	With these notations, the leading order contribution of the generator reads 
	\begin{align}
		\mathbb{H}  = \dG \big( h^{\l,\om}+  { U\vert \varphi_t \vert^2} + U \widetilde{K}_{1,t} \big) + \frac{U}{2} \sum_{ x \in\Zb^d}  \big[ \widetilde{K}_{2,t} (x) b_x^*b_x^* + \overline{\widetilde{K}}_{2,t} (x) b_x b_x \big] \label{211}.
	\end{align}
	Here, with $J: \ell^2( \Zb^d) \rightarrow \ell^2( \Zb^d)$ denoting the anti-linear operator $Jf = \overline{f}$, we write
	\begin{align}
		\label{tKdef}
		\widetilde{K}_{1,s} =& q_s K_{1,s} q_s,\quad \widetilde{K}_{2,s} = (Jq_sJ) K_{2,s} q_s, \\
		K_{1,s}(x) =& \varphi_s (x) \overline{\varphi}_s (x), \quad K_{2,s} (x) = \varphi_s (x) \varphi_s (x) . \label{tkDef2}
	\end{align}
	The remainder terms will be sub-leading for our analysis in the large particle limit and are given by 
	\begin{align}
		\label{def:Ri}
		\mathcal{R}_{N,t}^{(1)} =& \frac{U}{2} \dG \big( q_t \big[\vert \varphi_t \vert^2 \varphi_t + \widetilde{K}_{1,t} {-\mu_t}\big] q_t \big)  \frac{1- \mathcal{N}^+ (t)}{N} + U {\frac{\mathcal{N}^+(t) }{\sqrt{N}} b(q_t \vert \varphi_t \vert^2 \varphi_t )} + {\rm h.c.},  \ \\
		\mathcal{R}_{N,t}^{(2)} =& \frac{U}{\sqrt{N}} \sum_{x \in \Zb^d} \varphi_t (x) a^*(q_{t,x} ) a(q_{t,x} ) b (q_{t,x} ) + {\rm h.c.},   \\
		\mathcal{R}_{N,t}^{(3)} =& \frac{U}{N}\sum_{x \in \Zb^d}  a^*(q_{t,x} ) a^*(q_{t,x} )a(q_{t,x} ) a(q_{t,x} ) \; . \label{213}
	\end{align}
	In \eqref{def:Ri}, we set \begin{align}
		\label{mutDef}
		2\mu_t := \sum_{x \in \mathbb{Z}^d} \vert \varphi_t (x) \vert^2 \; \vert \varphi_t (y) \vert^2.
	\end{align} 
	We note that on the truncated Fock space it is convenient to work with the modified creation and annihilation operators $b^*(f),b(f)$ as defined in \eqref{def:b}, in contrast to the standard creation and annihilation operators, as the former leave the truncated Fock space $\mathcal{F}_{\perp \varphi_t}^{\leq N}$ invariant. However this comes with the price of modified canonical commutation relations  
	\begin{align}
		[b(f), b^*(g)] = \bigg( 1 - \frac{\mathcal{N}^+ (t)}{N} \bigg) \langle g,f \rangle - \frac{1}{N} a^*(g) a(f) , \quad [b^*(f), b^*(g)] = [b(f), b(g)] = 0 \; . 
	\end{align}
	
	Our results utilize estimates of (powers of) the number of excitations $\mathcal{N}^+ (t)$  w.r.t.~to the solution $\psi_{N,t}$ of the Schr\"odinger equation \eqref{eq:Schroe}. On the truncated Fock space $\mathcal{F}_{\perp \varphi_t}^{\leq N}$, we have $\mathcal{U}_{N,t} \mathcal{N}^+ (t) \mathcal{U}_{N,t}^* = \mathcal{N}^+ (t) = \mathcal{N}$, where $\cN$ denotes the global particle number    operator 
	\begin{align}
		\mathcal{N} := \sum_{z \in \Zb^d} n_z, \quad n_z = a^*_z a_z .
	\end{align}  See \cite{LRZa} for more details. Therefore, we have the relation
	\begin{align}
		\langle \psi_{N,t}, \mathcal{N}^+ (t) \psi_{N,t} \rangle = \langle \mathcal{W}_{N}(t;0) \mathcal{U}_{N,0} \psi_{N,0}, \mathcal{N} \mathcal{W}_{N}(t;0) \mathcal{U}_{N,0} \psi_{N,0} \rangle .
	\end{align}

	The next Lemma, proven in \cite{LRZa}, shows a bound on $\cN$. We introduce for any operator $\mathcal{A}$ on $\mathcal{F}_{\perp \varphi_t}^{\leq N}$ and any state $\mathcal{U}_{N,s} \psi_{N,s} \in \mathcal{F}_{\perp \varphi_s}^{\leq N}$  the shorthand notation 
	\begin{align}
		\label{def:Ast}
		\langle \mathcal{A} \rangle_{(t;s)} =  \langle \mathcal{W}_{N}(t;s) \mathcal{U}_{N,s} \psi_{N,s}, \mathcal{A} \mathcal{W}_{N}(t;s) \mathcal{U}_{N,s} \psi_{N,s} \rangle \; . 
	\end{align}

	\begin{lemma}
		[\cite{LRZa}, {Lem.~4.2}]
		\label{lem23} Let $\varphi_t,\,t\ge0$ be a solution to the Hartree equation \eqref{NLS} with initial data $\varphi_0 \in \ell^2$, $\norm{\varphi_0}_{\ell^2}\le1$.
		Then for $j=1,2,3$,  there exists  $C_* = C_* (j)>0$ such that 
		\begin{align}
			\langle (\mathcal{N} + 1)^j \rangle_{(t;s)} \leq \langle (\mathcal{N} + 1 )^j \rangle_{(s;s)} \; \exp \bigg( C_* \abs{U} \int_s^t d\si \| \varphi_\si \|_{\ell^\infty} \bigg),\quad 0\le s\le t. \label{eq:moment-bound}
		\end{align}
		
	\end{lemma}

	\section{Proof of \thmref{thm1}}\label{sec:proofthm1}
	We consider the fluctuation dynamics, defined in \eqref{def:W}, in the interaction picture. 
	
	Write $h=h^{\l,\om}=-\Lap+\l v_\om$ (see \eqref{HNdef}), $H_0:=\dG(h)\equiv \dG(-\Lap+\l v_\om)$, 
	\begin{align}
		\label{31}
		\tau _t :=&\tau_t^\Int\circ \tau _t ^{(0)},\quad \\
		\tau_t^\Int(A) :=&\mathcal{W}_N(0;t) e^{-itH_0}Ae^{itH_0} \mathcal{W}_N (t;0),\label{tauIntDef}\\
		\tau^{(0)}(A):=&e^{itH_0}Ae^{-itH_0}.
	\end{align} 
	It is shown in \cite{NS} that the map $t\mapsto \tau_t(A)$ is strongly differentiable for any bounded operator $A$ on $\Fl$.

	Our goal now is to prove dynamical control on the fluctuation dynamics $\tau_t$. There are two basic steps in the proof:
	\begin{enumerate}
		\item	Control $\tau_t^{(0)}$ through known results about 1-body dynamical localization;
		\item 	Extend to $\tau_t^\Int$ via commutator expansion.
	\end{enumerate}
	
	The key observation for both step is that $$e^{-itH_0}a_x=a(e^{-ith}\delta_x),$$ in which the $1$-body probability distribution evolving according to the $1$-body Schr\"odinger evolution $e^{-ith}$  satisfies SUDL \eqref{SUDL} by assumption. Specifically, setting $b=1$ in \eqref{SUDL} yields 
	\begin{align}
		\label{1partDL}
		\sup_t \abs{\br{e^{-ith}\delta_x,\delta_y}}\le C_1    e^{  \abs{x}}e^{-\g|x-y|},\quad x,\,y\in\Zb^d .
	\end{align} 
	We seek to establish a suitable  form of many-body localization via \eqref{1partDL} at the level of  fluctuations according to Steps 1-2 above. 
	
	\subsection{Estimate for $\tau_t^{(0)}$}
	We begin with the following bound on the free evolution $\tau^{(0)}$:
	\begin{lemma}\label{lem31}
		If \eqref{1partDL} holds, then there exists $C=C(d)>0$ such that for any $r>0$ and all $t$, the following operator inequality holds on $\Fl$: \begin{align}
			\label{MBL1}
			{ \tau_t^{(0)}(\N_{B_{r}})}\le  C \sum_{\abs{w}\le r} {\sum_z  e^{2\abs{w}-\g\abs{z-w}} {n_z}}.
		\end{align}
	\end{lemma}
	\begin{proof}
		Using the automorphism property $\tau_t^{(0)}(n_x)=
		\tau_t^{(0)}(a_x^* a_x)=\tau_t^{(0)}(a_x^* )\tau_t^{(0)}(a_x)$ and linearity of the mapping $f\mapsto a^
		*(f)$ (resp.\ anti-linearity of $f\mapsto a(f)$),
		we get
		\[
		\tau_t^{(0)}(\mathcal N_{B_r})
		=\sum_{|w|\leq r} a^*(e^{-ith}\delta_w)a(e^{ith}\delta_w)
		=\sum_{|w|\leq r}\sum_{x,y}  u_{x,w}(t) \bar u_{y,w}(t) a_x^* a_y,
		\]
		where $$u_{x,w}(t)=\langle e^{-ith}\delta_w,\delta_x\rangle$$ satisfies,  according to \eqref{1partDL},
		\begin{align}
			\label{uEst}
			\abs{u_{x,w}(t)}\le C_1e^{\abs{w}} e^{-\g\abs{x-w}} .
		\end{align}
		
		For any $\psi\in\Fl$, we compute, using triangle and Cauchy-Schwarz inequalities,
		\begin{align}
			\label{37}
			\abs{\wndel{\tau_t^{(0)}(\mathcal N_{B_r})  }} =&  \abs{\wndel{\sum_{|w|\leq r}\sum_{x,y}  u_{x,w}(t) \bar u_{y,w}(t) a_x^* a_y}}\notag\\
			\le&\sum_{|w|\leq r}\sum_{x,y}\abs{u_{x,w}(t)}\abs{u_{y,w}(t)}\abs{\wndel{a_x^*a_y}}\notag\\
			\le&\sum_{|w|\leq r}\sum_{x,y}\abs{u_{x,w}(t)}\abs{u_{y,w}(t)}\wndel{n_x}^{1/2}\wndel{n_y}^{1/2}.
		\end{align}
		For each fixed $w$, using estimate \eqref{uEst} for $\abs{u_{y,w}(t)}$ and Cauchy-Schwarz again    yields
		\begin{align}
			\label{}
			\sum_y\abs{u_{y,w}(t)}\wndel{n_y}^{1/2}\le &	 C_1e^{\abs{w}}\sum_y e^{-\g\abs{y-w}} \wndel{n_y}^{1/2}\notag\\
			\le& C_1e^{\abs{w}}\del{\sum_y e^{-\g\abs{y-w}} }^{1/2} \del{\sum_y  e^{-\g\abs{y-w}}\wndel{n_y}}^{1/2}, 
		\end{align}
		and a similar bound holds for $\sum_x\abs{u_{x,w}(t)}\wndel{n_x}^{1/2}$. Plugging these back to \eqref{37} leads to 
		\begin{align}
			\label{39'}
			\abs{\wndel{\tau_t^{(0)}(\mathcal N_{B_r})  }} \le C_1^2C_{d} \sum_{\abs{w}\le r}  e^{2\abs{w}} {\sum_z  e^{-\g\abs{z-w}}\wndel{n_z}} ,
		\end{align}
		where  
		\begin{align} 
			\label{CrDef}
			C_{d}=  \sum_{z\in \Zb^d} e^{-\g\abs{z}}
		\end{align} 
		is bounded independent of $\g$ for any $\g\ge1$. 
		Setting $C=C_1^2C_d$ in \eqref{39'}  gives the desired result.
		
	\end{proof}

	\subsection{Estimate for $\tau_t^\Int$}
	
	The interaction dynamics is estimated using the following commutator bound:
	\begin{lemma}\label{lemMain}
		
		Let $T:\Zb^d \rightarrow \mathbb{R} $. Then, we have   for any $\psi\in\Fl$, with $\| \psi \|=1 $ and with the notation $\psi_t := e^{-itH_0} \psi $,
		\begin{align}\label{412}
			& \abs{\wndel{\big[H_{\rm int} (t), \sum_{z \in \Zb^d} T(z) n_z \big]}} \notag \\
			&\leq  C  \sum_{x,y,z \in \Zb^d}  \vert T(z) \vert \; \vert  u_{x,z} (t) \vert \; \vert u_{y,z} (t) \vert \notag \\
			&\hspace{0.5cm} \times  \bigg( \vert \varphi_t (y) \vert^2 \; \| n_x^{1/2} \psi_t \| \; \|( n_y + 1)^{1/2} \psi_t \|\notag \\
			& \hspace{1cm} +  \frac{1}{\sqrt{N}} \bigg[ \vert \varphi_t (x) \vert^3  \| \mathcal{N} \psi_t \|    \; \|n_y^{1/2} \psi_t \| + \vert \varphi_t (x) \vert \; \|n_y^{1/2} \psi_t \| \; \|(n_x+1) \psi_t \| +  \vert \varphi_t (x) \vert \| n_y^{1/2} \psi_t \| \bigg]  \notag \\
			& \hspace{1cm} + \frac{1}{N}  \|(n_x(n_x+1))^{1/2} \psi_t \|\; \| n_y^{1/2} ( \mathcal{N}+1)^{1/2} \psi_t \|  \bigg) ,
		\end{align}
		where %
		\begin{align}
			\label{HintDef}
			H_{\rm int} (t) = e^{itH_0} \big( \mathcal{L}_N (t) - H_0 \big) e^{-it H_0},
		\end{align} and 
		\begin{align}
			\label{def:u}
			u_{x,z} (t) = \langle \delta_z, e^{-ith} \delta_x \rangle.
		\end{align} 
	\end{lemma}

	The proof of this lemma is postponed to Sect.~\ref{sec44}, and 
	we first derive the desired bound on interaction dynamics from \eqref{412}.

	\begin{lemma}\label{lem42}
		Let $r>0$, $R\ge2r$, and $\rho=R-r$. 
		Let $\norm{\varphi_0}_{\ell^2}=1$, $\varphi_0(x)=0$ for $\abs{x}\ge R$, and assume \ref{C2} holds. Then, given any $M,\,\eps>0$, there exist  $\l_*=\l_*(M,\,\eps, \rho)>0$  such that for $\l\ge\l_*$ and all $0\le t \le \rho/\eps $, the following operator inequality holds on $\Fl$:
		\begin{align}
			\label{MBL2}
			i {\big[H_{\rm int} (t), \sum_{z }\sum_{\abs{w}\le r} e^{2\abs{w}-\g\abs{z-w}} n_z \big]} \le &   C_{r,d}\del{ e^{-\frac{M}{\tilde \eps}\rho}{(\cN+1)}+ \frac{1}{N}{(\cN+1)^2}}.
		\end{align}
		Here and in what follows, we set $\tilde \eps := 3\eps$. 
	\end{lemma}
	\begin{proof}
		Take 	$\psi\in\Fl$ with $\norm{\psi}=1$, and write  $\psi_t := e^{-itH_0} \psi $.
		As the starting point, we apply commutator estimate \eqref{412} with the choice
		$$T(z)=\sum_{\abs{w}\le r} e^{2\abs{w}-\g\abs{z-w}} $$
		to obtain
		\begin{align}
			\label{415}
			&\abs{\wndel{\big[H_{\rm int} (t), \sum_{z }\sum_{\abs{w}\le r} e^{2\abs{w}-\g\abs{z-w}} n_z \big]}} \notag\\\le& C \sum_z\sum_{\abs{w}\le r}e^{-\frac{\g}{2}\abs{z-w}} ( G (z,w) + F(z,w)+ H(z,w)),
		\end{align}
		where 

		\begin{align}
			\label{GzDef}
			G(z,w):=& e^{2\abs{w}-\frac{\g}{2}\abs{z-w}} 	\sum_{x,y}\; \vert  u_{x,z} (t) \vert \; \vert u_{y,z} (t) \vert \vert \varphi_t (y) \vert^2  \|n_x^{1/2}  e^{-i t H_0}\psi \|\; \|(n_y + 1)^{1/2} e^{-i t H_0}\psi \| ,\\
			F(z,w):=& \frac{e^{2\abs{w}-\frac{\g}{2}\abs{z-w}} }{\sqrt{N}}	\sum_{x,y}\; \vert  u_{x,z} (t) \vert \; \vert u_{y,z} (t) \vert \notag\\
			&\times\del{\vert \varphi_t (x) \vert^3  \| \mathcal{N} \psi_t \|    \; \|n_y^{1/2} \psi_t \| + \vert \varphi_t (x) \vert \; \|n_y^{1/2} \psi_t \| \; \|(n_x+1) \psi_t \|  }  ,\label{Fdef}\\
			H(z,w):=&\frac{e^{2\abs{w}-\frac{\g}{2}\abs{z-w}} }N \sum_{x,y} \vert  u_{x,z} (t) \vert \; \vert u_{y,z} (t) \vert   \|(n_x(n_x+1))^{1/2} \psi_t \|\; \| n_y^{1/2} ( \mathcal{N}+1)^{1/2} \psi_t \|.\label{HzDef}
		\end{align}

		In what follows, we derive uniform bounds for $z\in\Zb^d,\,\abs{w}\le r$ on the contributions of $G,\,F$ and $H$ separately.

		\textbf{Bound on $G(z,w)$.}  
		
		Using the relation $\norm{(n_x+a)^{1/2}u}^2=\br{u,(n_x+a)u}$ for $a\ge0$ and Cauchy-Schwarz inequality, we find
		
		\begin{align}
			\label{418}
			G(z,w)	
			\le&   G_1(z,w)^{1/2}G_2(z,w)^{1/2}, \\
			G_1(z,w):=& e^{2\abs{w}-\frac{\g}{2}\abs{z-w}} {\sum_{x,y}\; \vert  u_{x,z} (t) \vert \; \vert u_{y,z} (t) \vert \abs{\varphi_t(y)}^2\wtdel{n_x}} ,
			\\
			G_2(z,w):=&e^{2\abs{w}-\frac{\g}{2}\abs{z-w}} {\sum_{x,y}\; \vert  u_{x,z} (t) \vert \; \vert u_{y,z} (t) \vert \abs{\varphi_t(y)}^2\wtdel{(n_y+1)}}  .
		\end{align}
		
		We start with bounding $G_2$. The bound for $G_1$ is simpler and is obtained similarly.

		To begin with, by  Fubini's theorem, we have
		\begin{align}
			\label{4200}
			G_2(z,w)
			=&e^{2\abs{w}-\frac{\g}{2}\abs{z-w}}  \del{\sum_x  \abs{u_{x,z} (t)} }\del{\sum_y \abs{u_{y,z} (t)}\abs{\varphi_t(y)}^2\wtdel{ (n_y+1) } }.
		\end{align}
		Since $\wtdel{(n_y+1)}\le\wtdel{(\cN+1)}$ for every $y$ and $[H_0,\cN]=0$, the second term is generously bounded as
		\begin{align}
			\label{420}
			&\sum_y \abs{u_{y,z} (t)}\abs{\varphi_t(y)}^2\wtdel{ (n_y+1) }\notag\\
			\le& \wtdel{(\cN+1)}	\sum_y \abs{u_{y,z} (t)}\abs{\varphi_t(y)}^2\notag\\
			=			 & \wndel{(\cN+1)}	\sum_y \abs{u_{y,z} (t)}\abs{\varphi_t(y)}^2.
		\end{align}
		Thus \eqref{4200} becomes
		\begin{align}
			\label{4201}
			G_2(z,w)\le& \wndel{(\cN+1)}\Phi(z,w),\\
			\label{PhiDef}\Phi(z,w):=&e^{2\abs{w}-\frac{\g}{2}\abs{z-w}} \del{\sum_x  \abs{u_{x,z} (t)} }\del{\sum_y \abs{u_{y,z} (t)}\abs{\varphi_t(y)} }.
		\end{align}
		To get \eqref{PhiDef}, which will be convenient for future reference, we have used that $\norm{\varphi_t}_{\ell^\infty}\le \norm{\varphi_0}_{\ell^2}\equiv1$. 
		
		We now claim the following:
		\begin{lemma}\label{lem44}Let the assumptions of \lemref{lem42} hold.
			Let $M,\,\eps>0$. Then there exists $\l_*=\l_*(M,\eps)>0$ such that for every $\l\ge\l_*$,  
			\begin{align}
				\label{PhiEst}
				\sup_{z\in\Zb^d,\abs{w}\le r}	\Phi(z,w) \le C_{d} e^{-M\rho/  \eps},\quad 0\le t \le \rho/\eps.
			\end{align}
		\end{lemma}
		The proof of this lemma is rather technical and is postponed to \secref{secPfLem44}. 
		Below we show how to use \eqref{PhiEst} to derive bounds for $G$ and $F$.

		Applying \eqref{PhiEst} to the bound \eqref{4201}, we find that 
		\begin{align}
			\label{G2est}
			G_2(z,w)\le C_{d}e^{-\frac{M}{\tilde \eps}\rho} \wndel{(\cN+1)}.
		\end{align}
		To bound $G_1$, we use, in place of \eqref{4200}, that
		\begin{align}
			\label{433}
			G_1(z,w)=&e^{2\abs{w}-\frac{\g}{2}\abs{z-w}} \del{\sum_x \abs{u_{x,z} (t)} \wtdel{n_x } }\del{\sum_y  \abs{u_{y,z} (t)} \abs{\varphi_t(y)}^2}.
		\end{align}
		A straightforward adaption of the estimates for $G_2$ above shows that 
		\begin{align}
			\label{4202}
			G_1(z,w)\le \wndel{\cN}\Phi(z,w).
		\end{align}
		This, together with \lemref{lem44}, shows that  for $\l$ sufficiently large,
		\begin{align}
			\label{G1est}
			G_1(z,w)\le C_{d}e^{-\frac{M}{\tilde \eps}\rho} \wndel{\cN}.
		\end{align}
		
		Finally, to conclude the bound on $G(z,w)$, we combine
		\eqref{G1est}, \eqref{G2est} in \eqref{418} to obtain
		\begin{align}
			\label{mainEst0}
			\sup_{z\in\Zb^d,\abs{w}\le r}G(z,w)\le C_{d}e^{-\frac{M}{\tilde \eps}\rho} \wndel{(\cN+1)}. 
		\end{align}
		
		\textbf{Bound on $F(z,w)$.}  
		Recall 
		\begin{align}
			\label{}
			F(z,w):=& \frac{e^{2\abs{w}-\frac{\g}{2}\abs{z-w}} }{\sqrt{N}}	\sum_{x,y}\; \vert  u_{x,z} (t) \vert \; \vert u_{y,z} (t) \vert \notag\\
			&\times\del{\vert \varphi_t (x) \vert^3  \| \mathcal{N} \psi_t \|    \; \|n_y^{1/2} \psi_t \| + \vert \varphi_t (x) \vert \; \|n_y^{1/2} \psi_t \| \; \|(n_x+1) \psi_t \|  }  ,\notag
		\end{align}
		$F$ has similar structure as $G$ and is controlled similarly. Indeed, by the Cauchy-Schwarz inequality and the fact that $\sup_x(n_x+1)^2\le (\cN+1)^2$, similar to \eqref{418} we have
		\begin{align}
			\label{Fest}
			F(z,w)\le&  2 F_1 (z,w)^{1/2}F_2(z,w)^{1/2} , \\
			F_1 (z,w):=&	  e^{2\abs{w}-\frac{\g}{2}\abs{z-w}}  {\sum_{x,y}\; \vert  u_{x,z} (t) \vert \; \vert u_{y,z} (t) \vert \abs{\varphi_t(x)} \frac{\wtdel{(\cN+1)^2}}{N}} , \notag\\
			F_2 (z,w):=&	 e^{2\abs{w}-\frac{\g}{2}\abs{z-w}}  {\sum_{x,y}\; \vert  u_{x,z} (t) \vert \; \vert u_{y,z} (t) \vert \abs{\varphi_t(x)}  {\wtdel{n_y}}  }.\notag
		\end{align}
		Similar to the bound for $G_1$ and $G_2$ (see \eqref{4202} and\eqref{4201}), we have 
		$$
		F_1(z,w)\le \frac{\wndel{(\cN+1)^2}}{N}\Phi(z,w),\quad F_2(z,w)\le \wndel{\cN}\Phi(z,w).
		$$
		These, together with \lemref{lem44},    yields
		\begin{align}
			\label{mainEst1}
			\sup_{z\in\Zb^d,\abs{w}\le r}F(z,w)\le C_{d}e^{-\frac{M}{\tilde \eps}\rho} \frac{\wndel{(\cN+1)^2}^{3/4}}{\sqrt{N}}. 
		\end{align}

		\textbf{Bound on $H(z,w)$.}  
		Recall
		\begin{align}
			\label{}
			H(z,w):=&\frac{e^{2\abs{w}-\frac{\g}{2}\abs{z-w}} }N \sum_{x,y} \vert  u_{x,z} (t) \vert \; \vert u_{y,z} (t) \vert   \|(n_x(n_x+1))^{1/2} \psi_t \|\; \| n_y^{1/2} ( \mathcal{N}+1)^{1/2} \psi_t \|.\notag
		\end{align}
		Compared with $G$ and $F$, the key difference is that $H(z,w)$ does not carry any factor of $\varphi_t(x)$ but instead an overall factor of $1/N$. Thus we use the trivial bounds $ e^{-\frac{\g}{2}\abs{z-w}},\,\abs{u_{x,z}(t)},\,\abs{u_{y,z}(t)}\le1$ and H\"older's inequality to obtain 
		\begin{align}
			\label{mainEst2}
			\sup_{z\in\Zb^d,\abs{w}\le r}	H(z,w)\le & \frac{e^{2r}}{N}\del{\sum_{x}\wtdel{n_x(n_x+1)}}^{1/2}\del{\sum_{y}\wtdel{n_y(\cN+1)}}^{1/2} \notag\\
			\le&  \frac{e^{2r}}{N}\wndel{(\cN+1)^2}.
		\end{align}

		\textbf{Completing the proof.} Using \eqref{mainEst0}, \eqref{mainEst1}, and \eqref{mainEst2} to bound $G,\,F,\,H$ respectively in the r.h.s.~of \eqref{415}, we arrive at	\begin{align}
			\label{}
			& \abs{\wndel{\big[H_{\rm int} (t), \sum_{z }\sum_{\abs{w}\le r} e^{2\abs{w}-\g\abs{z-w}} n_z \big]}}\\
			\le& C \sum_z\sum_{\abs{w}\le r}e^{-\frac{\g}{2}\abs{z-w}}\notag\\
			&\times \del{\sup_{z\in\Zb^d,\abs{w}\le r}	G(z,w)+\sup_{z\in\Zb^d,\abs{w}\le r}	F(z,w)+\sup_{z\in\Zb^d,\abs{w}\le r}	H(z,w)}\notag\\
			\le &   C_{d}\del{ e^{-\frac{M}{\tilde \eps}\rho}\wndel{(\cN+1)}+ e^{-\frac{M}{\tilde \eps}\rho} \frac{\wndel{(\cN+1)^2}^{3/4}}{\sqrt{N}}}+\frac{e^{2r}}{N}\wndel{(\cN+1)^2}.
		\end{align}
		Since $\eps$, $\psi$ are arbitrary, and $\cN\le N$ on $\Fl$, the desired result follows from here.

	\end{proof}

	\subsection{Completing the proof of \thmref{thm1}}\label{secPfThm1}
	Within this subsection, we let $\xi_0=\cU_{N,0}\psi_{N,0}$ stand for the initial fluctuation  and write, for any bounded operator on $\Fl$,
	$$\br{A}_t = \br{\xi_0 ,\tau_t(A)\xi_0},$$   
	where, recall, $\tau_t(\cdot)$ stands for the fluctuation dynamics \eqref{31}.
	Compared with definition  \eqref{def:Ast}, we have
	\begin{align}
		\label{rel1}
		\br{\cN_X}_t\equiv \br{\cN_X}_{(t;0)},\quad X\subset\Zb^d.
	\end{align}

	First, by relation \eqref{31}, we have
	\begin{align}
		\label{51}
		\br{\cN_{B_r}}_t = \br{\xi_0, \tau_t^\Int \circ \tau_t^{(0)}(\cN_{B_r})\xi_0}.
	\end{align}
	Using \lemref{lem31} to bound the free evolution $\tau_t^{(0)}(\cN_{B_r})$, we find 
	\begin{align}
		\label{52}
		\br{\xi_0, \tau_t^\Int \circ \tau_t^{(0)}(\cN_{B_r})\xi_0}\le C_d \xzdel{\tau_t^\Int\del{\sum_{\abs{w}\le r} {\sum_{z\in\Zb^d}  e^{2\abs{w}-\g\abs{z-w}}{n_z}}}}.
	\end{align}
	
	Next, we claim that the interaction dynamics $\tau_t^\Int$ defined in \eqref{tauIntDef} satisfies the Heisenberg equation
	\begin{align}
		\label{HeisE}
		\di_t\tau_t^\Int(A) = \tau_t^\Int(i[H_\Int(t),A]),
	\end{align} 
	where, recall, $H_\Int(t)$ is the interaction Hamiltonian given by \eqref{HintDef}. Indeed, using relation \eqref{Weq}, a direct calculation shows that
	$\di_t(e^{itH_0}\W_N(t;0))=e^{itH_0}(iH_0-i\cL_N(t))\W_N(t)=-iH_\Int(t) e^{itH_0}\W_N(t;0)$, and similarly that $\di_t(\W_N(0;t)e^{-itH_0})=i\W_N(0;t) e^{-itH_0}H_\Int(t)$. Consequently, eq.~\eqref{HeisE} follows.

	Using   the Heisenberg equation \eqref{HeisE} and the fundamental theorem of calculus, for every $z\in\Zb^d$, we have the identity \begin{align}
		\label{}
		\xzdel{\tau_t^\Int(n_z)} - \xzdel{n_z} =   \int _0^t \xzdel{\tau_s ^{\Int}(i[H_\Int(t),n_z])}\,ds. 
	\end{align}
	This, together with \lemref{lem42}, the fact that $[H_0,\cN]=0$, and relation \eqref{51},   implies
	\begin{align}
		\label{}
		&\xzdel{\tau_t^\Int\del{ \sum_{\abs{w}\le r} {\sum_{z\in\Zb^d}  e^{2\abs{w}-\g\abs{z-w}}{n_z}}}} -\xzdel{\del{ \sum_{\abs{w}\le r} {\sum_{z\in\Zb^d}  e^{2\abs{w}-\g\abs{z-w}}{n_z}}}}
		\label{55}\\
		=&\int_0^t \xzdel{\tau_s^\Int \del{i {\big[H_{\rm int} (t), \sum_{z }\sum_{\abs{w}\le r} e^{2\abs{w}-\g\abs{z-w}} n_z \big]} }}\,ds\notag\\
		\le& C_{r,d}\int_0^t \xzdel{\tau_s^\Int \del{e^{-M\rho/\eps}(\cN+1)+\frac{(\cN+1)^2}{N}} }\,ds\notag\\
		=& C_{r,d}\int_0^t \left\langle{  {e^{-M\rho/\eps}(\cN+1)+\frac{(\cN+1)^2}{N}} }\right\rangle_s\,ds.\notag
	\end{align}
	
	By the   global fluctuation bound from \lemref{lem23}, we have for some $\tilde K=\tilde K(U)>0$ that $$\br{(\cN+1)^j }_t\le  e^{\tilde Kt}\br{(\cN+1)^j }_0,\qquad j=1,2.$$
	Furthermore, by the assumption \eqref{psi0Cond}  on the initial state, we have $\br{\psi_0,\cN_{B_R}\psi_0}=0$. These, together with  H\"older's inequality, show that the integral in line \eqref{55} is bounded as 
	\begin{align}
		\label{56''}
		&\int_0^t \left\langle{  {e^{-M\rho/\eps}(\cN+1)+\frac{(\cN+1)^2}{N}} }\right\rangle_s\,ds
		\le Ct \del{e^{\tilde K t-M\rho/\eps }+\frac{e^{\tilde K t}}{N}}\br{(\cN_{B_R^\cp}+1)^2 }_0.
	\end{align}
	Similarly,  again due to  assumption \eqref{psi0Cond}, the substracted term in the l.h.s.~of \eqref{55} is bounded as
	\begin{align}
		\label{56'''}
		&\xzdel{\del{ \sum_{\abs{w}\le r} {\sum_{z\in\Zb^d}  e^{2\abs{w}-\g\abs{z-w}}{n_z}}}}\notag\\
		=&\xzdel{\del{ \sum_{\abs{w}\le r} {\sum_{\abs{z}\ge R}  e^{2\abs{w}-\g\abs{z-w}}{n_z}}}}\notag\\\le& C_{r,d}e^{-\g\rho}\br{\cN_{B_R^\cp}}_0.
	\end{align}
	
	Plugging \eqref{56''}---\eqref{56'''}  back to \eqref{55} yields
	\begin{align}
		\label{}
		&\xzdel{\tau_t^\Int\del{ \sum_{\abs{w}\le r} {\sum_{z\in\Zb^d}  e^{2\abs{w}-\g\abs{z-w}}{n_z}}}}\notag\\&\le C_{r,d}\sbr{t\del{e^{\tilde K t-M\rho/\eps }+\frac{e^{\tilde K t}}{N}}\br{(\cN_{B_R^\cp}+1)^2 }_0+{e^{-\g\rho}\br{\cN_{B_R^\cp}}_0}}.
	\end{align}
	This, together with \eqref{52},   the fact that $te^{\tilde K t}\le e^{ (\tilde K +1)t}$, and the choice %
	\begin{align}
		\label{Mcond}
		M:= K:= \tilde K+1,
	\end{align}  yields the desired estimate \eqref{slowPE}. This completes the proof of \thmref{thm1}.

	\subsection{Proof of \lemref{lemMain}}\label{sec44}
	
	We write 
	\begin{align}
		i \big[H_{\rm int} (t), \sum_{z \in {\Zb^d}} T(z) n_z \big] =& e^{itH_0} \big[\mathcal{L}_N(t) - H_0, \sum_{z \in {\Zb^d}} T(z) e^{-it H_0 }n_z e^{it H_0 }\big] e^{-i t H_0}  \label{eq:comm-Hint-1}
	\end{align}
	and observe, using the canonical commutation relations $[a(f), a^*(g)] = \langle f, g \rangle$, that
	\begin{align}
		e^{-it H_0 }n_z e^{it H_0 } = e^{-it H_0 } a_z^* a_z e^{it H_0 } = e^{-it H_0 } a^*( \delta_z) a( \delta_z) e^{it H_0 } = a^*( e^{-ith} \delta_z) a( e^{-ith} \delta_z) 
	\end{align}
	where we introduced the notation $\delta_z = \delta ( \cdot - z)$ that we can further rewrite as 
	\begin{align}
		e^{-it H_0 }n_z e^{it H_0 } 
		=  \sum_{x,y \in {\Zb^d}}   \langle \delta_z, e^{-ith} \delta_x \rangle \; \overline{\langle \delta_z, e^{-ith} \delta_y \rangle} a_x^*a_y = \sum_{x,y \in {\Zb^d}}  u_{x,z} (t) \overline{u_{y,z} (t)} a_x^*a_y 
	\end{align}
	with the definition $u_{x,z} (t) = \langle \delta_z, e^{-ith} \delta_x \rangle$. Plugging this back into \eqref{eq:comm-Hint-1}, we are left with computing 
	\begin{align}
		i \big[H_{\rm int} (t),  & \sum_{z \in {\Zb^d}} T(z) n_z \big] \notag \\
		=& \sum_{x,y,z \in {\Zb^d}} T(z)  u_{x,z} (t)  \; \overline{u_{y,z} (t)} e^{itH_0} \big[\mathcal{L}_N(t) - H_0, a_x^*a_y\big] e^{-i t H_0}  
	\end{align} 
	and in the following we estimate the commutator on the r.h.s. For this, we recall the decomposition of the generator $\mathcal{L}_N(t) = \mathbb{H} + \sum_{j=1}^3 \mathcal{R}_{N,t}^{(j)}$ in \eqref{211} resp. \eqref{213} and thus arrive at 
	\begin{align}
		i \big[H_{\rm int} (t),  & \sum_{z \in {\Zb^d}} T(z) n_z \big] \notag  \notag \\
		=& \sum_{x,y,z \in {\Zb^d}} T(z)   u_{x,z} (t)  \; \overline{u_{y,z} (t)}  \; e^{itH_0} \big[\mathbb{H} - H_0, a_x^*a_y\big] e^{-i t H_0} \notag \\
		&+ \sum_{x,y,z \in {\Zb^d}} T(z)   u_{x,z} (t)  \overline{u_{y,z} (t)}  \; e^{itH_0} \big[\sum_{i=1}^3 \mathcal{R}_{N,t}^{(j)}, a_x^*a_y\big] e^{-i t H_0} \label{eq:comm-Hint-1'}
	\end{align} 
	and we are left with estimating the two terms of the r.h.s. of \eqref{eq:comm-Hint-1'}. We start with the first one, for which we compute 
	\begin{align}
		\big[\mathbb{H}   - H_0, a_x^*a_y\big] =& \frac{1}{2} \sum_{w \in {\Zb^d} } \bigg[ \widetilde{K}_{2,t} (w) b_w^*b_w^*+ \overline{\widetilde{K}}_{2,t} (w) b_wb_w, a_x^*a_y ] \notag \\
		=& \widetilde{K}_{2,t} (y)  b_x^*b_y^*  -  \overline{\widetilde{K}}_{2,t} (x) b_x b_y 
	\end{align}
	yielding with the notation $\psi_t := e^{-iH_0 t} \psi $ for any $\psi \in \mathcal{F}_{\perp \varphi}^{\leq N}$ to 
	\begin{align}
		\vert \langle  &\psi, 	\sum_{x,y,z \in {\Zb^d}}   T(z)   u_{x,z} (t)  \; \overline{u_{y,z} (t)} e^{itH_0} \big[\mathbb{H} - H_0, a_x^*a_y\big] e^{-i t H_0} \psi \rangle \vert \notag \\
		& \leq  \sum_{x,y,z \in {\Zb^d}} \vert T(z) \vert \; \vert  u_{x,z} (t) \vert \; \vert u_{y,z} (t) \vert \bigg( \|\widetilde{K}_{2,t} (y) \|\; \|b_x   \psi_t \|\; \|b_y^* \psi_t \| + \|\widetilde{K}_{2,t} (x) \|\; \|b_y  \psi_t \|\; \|b_x^* \psi_t \| \bigg) \notag \\
		& \leq  \sum_{x,y,z \in {\Zb^d}} \vert T(z) \vert \; \vert  u_{x,z} (t) \vert \; \vert u_{y,z} (t) \vert \bigg( \vert \varphi_t (y) \vert^2  \|n_x^{1/2}  \psi_t \|\; \|(n_y + 1)^{1/2} \psi_t \| \notag \\
		&\hspace{6.5cm} + \vert \varphi_t (x) \vert^2\; \|n_y^{1/2}  \psi_t \|\; \| (n_x + 1)^{1/2} \psi_t \| \bigg) \; \label{eq:estimate-R0}
	\end{align}
	that is of the desired form.

	We continue with the second term of the r.h.s. of \eqref{eq:comm-Hint-1'} for which we consider the single contributions of the remainder term separately. For the first we get with the notation $G_t= q_t \big[ \vert \varphi_t \vert^2 \varphi_t + \widetilde{K}_{1,t} \big] q_t $ and $g_t = q_t \vert \varphi_t \vert^2 \varphi_t $ 
	\begin{align}
		\big[ \mathcal{R}_{N,t}^{(1)}, a_x^*a_y] = \big( G_t(x) a_x^* a_y - G_t(y) a_x^*a_y \big) \frac{1-\mathcal{N}_+}{N} -  \frac{\mathcal{N}_+}{\sqrt{N}} g_t(x) b_y  + g_t(y) b_x^* \frac{\mathcal{N}_+}{\sqrt{N}}  \; . 
	\end{align}
	Since  $\vert G_t (x) \vert \leq  C \vert \varphi_t (x) \vert^2 $ we get on the one hand for any $\psi \in \mathcal{F}_{\perp \varphi}^{\leq N}$ 
	\begin{align}
		\vert \langle & \psi_t, \sum_{x,y,z \in {\Zb^d}} T(z) u_{x,z} (t) \overline{u_{y,z} (t)} \big( G_t (x) - G_t (y) \big) \frac{1-\mathcal{N}}{N} a_x^*a_y \psi_t \rangle \notag \\
		\leq& C  \sum_{x,y,z \in {\Zb^d} } \vert T(z) \vert \; \vert u_{x,z} (t) \vert \; \vert u_{y,z} (t) \vert \; \big( \vert \varphi_t (x)\vert^2 + \vert \varphi_t (y) \vert^2 \big) \| a_x \psi_t \| \; \|a_y \psi_t \| 	\notag \\
		\leq& C \sum_{x,y,z \in {\Zb^d}} \vert T(z) \vert \; \vert u_{x,z} (t) \vert \; \vert u_{y,z} (t) \vert \;  \big( \vert \varphi_t (x) \vert^2 + \vert \varphi_t (y) \vert^2 \big) \| n_x^{1/2} \psi_t \| \; \|n_y^{1/2} \psi_t \| \; . \label{eq:estimate-R1,1}
	\end{align}
	On the other hand, since $\vert g_t (x) \vert \leq C \vert \varphi_t ( x) \vert^3 $, we find 
	\begin{align}
		\vert 	\langle &  \psi_t, \sum_{x,y,z \in {\Zb^d}}   T(z)   u_{x,z} (t)  \; \overline{u_{y,z} (t)} \bigg(  \frac{\mathcal{N}}{\sqrt{N}} g_t (x) b_y + g_t (y) b_x^* \frac{\mathcal{N}}{\sqrt{N}} \bigg) \psi_t \rangle \vert \notag \\
		& \leq \frac{C}{\sqrt{N}}   \sum_{x,y,z \in {\Zb^d}}   \vert T(z) \vert    \; \vert u_{x,z} (t) \vert \; \vert u_{y,z} (t) \vert \; \vert \varphi_t (x) \vert^3 \big\| \mathcal{N} \psi_t \|\; \|b_y \psi_t \| \notag \\
		&\quad + \frac{C}{\sqrt{N}} \sum_{x,y,z \in {\Zb^d}} \vert T(z) \vert \; \vert u_{x,z} (t) \vert \; \vert u_{y,z} (t) \vert \; \vert \varphi_t (y) \vert^3 \|\mathcal{N} \psi_t \| \; \| b_x \psi_t \| \notag \\
		& \leq \frac{C}{\sqrt{N}}  \sum_{x,y,z \in {\Zb^d}} \vert T(z) \vert \; \vert u_{x,z} (t) \vert \; \vert u_{y,z} (t) \vert \; \| \mathcal{N} \psi_t \|  \bigg( \vert \varphi_t (x) \vert^3 \, \| n_y^{1/2} \psi_t \|  +  \vert \varphi_t (y) \vert^3 \; \|n_x^{1/2} \psi_t \| \bigg) \label{eq:estimate-R1,2}.
	\end{align}

	For the second term of the remainder in \eqref{213} we recall that on $\mathcal{F}_{\perp \varphi_t}^{\leq N}$ we have 
	\begin{align}
		\mathcal{R}_{N,t}^{(2)} = \frac{1}{\sqrt{N}} \sum_{w \in {\Zb^d} } \varphi_t (w) a_w^* a_w b_w + {\rm h.c.}
	\end{align}
	and thus 
	\begin{align}
		\big[\mathcal{R}_{N,t}^{(2)}, a_x^*a_y] =& \frac{1}{\sqrt{N}} \sum_{w \in {\Zb^d}} \bigg( \varphi_t (w) \big[ a_w^*a_wb_w, a_x^*a_y] + \overline{\varphi}_t (w) \big[ b_w^*a_w^*a_w, a_x^*a_y]  \bigg)  \notag \\
		=& \frac{1}{\sqrt{N}} \sum_{w \in {\Zb^d}} \varphi_t (w)\bigg(  \big[ a_w^*a_w, a_x^*a_y]b_w + a_w^*a_w\big[b_w , a_x^*a_y]  \bigg) \notag \\
		&+  \frac{1}{\sqrt{N}} \sum_{w \in {\Zb^d}} \overline{\varphi}_t (w)\bigg( b_w^* \big[ a_w^*a_w, a_x^*a_y] + \big[b_w^* , a_x^*a_y]  a_w^*a_w\bigg) \notag \\
		=& \frac{1}{\sqrt{N}} \bigg( \varphi_t (x) \;  a_x^* a_y b_x  -  \varphi_t (y) \; a_x^*a_y b_y+  \varphi_t (x) \; a_x^*a_x b_y  \bigg) \notag \\
		&+  \frac{1}{\sqrt{N}}\bigg( \overline{\varphi}_t (x) \;   b_x^*a_x^*a_y -  \overline{\varphi}_t (y)  \; b_y^* a_x^* a_y - \overline{\varphi}_t (y) \;  b_x^*  a_y^*a_y \bigg) \notag\\
		=& \frac{1}{\sqrt{N}} \bigg( \varphi_t (x) \;  b_x a_x^* a_y - \varphi_t (x) \; b_y  -  \varphi_t (y) \; a_x^*a_y b_y+  \varphi_t (x) \; a_x^*a_x b_y  \bigg) \notag \\
		&+  \frac{1}{\sqrt{N}}\bigg( \overline{\varphi}_t (x) \;   b_x^*a_x^*a_y -  \overline{\varphi}_t (y)  \; a_x^* a_y b_y^* -\overline{\varphi}_t (y)  \, b_x^* - \overline{\varphi}_t (y) \;  b_x^*  a_y^*a_y \bigg).
	\end{align}
	Since $[b_w^*, a_x^*a_y] = - \delta (w-y) b_x^* $, we arrive at 
	\begin{align}
		\big[\mathcal{R}_{N,t}^{(2)}, a_x^*a_y] =& \frac{1}{\sqrt{N}} \sum_{w \in {\Zb^d}} \bigg( \varphi_t (w) \big[ a_w^*a_wb_w, a_x^*a_y] + \overline{\varphi}_t (w) \big[ b_w^*a_w^*a_w, a_x^*a_y]  \bigg)  \notag \\
		=& \frac{1}{\sqrt{N}} \bigg( \varphi_t (x) \;  b_x a_x^* a_y - \varphi_t (x) \; b_y  -  \varphi_t (y) \; a_x^*a_y b_y+  \varphi_t (x) \; a_x^*a_x b_y  \bigg) \notag \\
		&+  \frac{1}{\sqrt{N}}\bigg( \overline{\varphi}_t (x) \;   b_x^*a_x^*a_y -  \overline{\varphi}_t (y)  \; a_x^* a_y b_y^* -\overline{\varphi}_t (y)  \, b_x^* - \overline{\varphi}_t (y) \;  b_x^*  a_y^*a_y \bigg) \label{467}
	\end{align}
	that yields for any $\psi \in \mathcal{F}_{\perp \varphi_t}^{\leq N}$ and recalling the notation $\psi_t = e^{-iH_0 t} \psi $ to 
	\begin{align}
		&   \big\vert   \langle \psi_t, \big[ \mathcal{R}_{N,t}^{(2)}, a_x^*a_y ] \psi_t \rangle \big\vert \notag \\
		& \quad  \leq   \frac{1}{\sqrt{N}} \bigg(  \vert \varphi_t (x) \vert \; \|a_x b_x^* \psi_t \| \; \|a_y  \psi \| + \vert \varphi_t (x) \vert \; \| b_y \psi_t \| \; \| \psi \| + \vert \varphi_t (y) \vert \; \| a_x \psi_t \| \; \| a_yb_y \psi_t \| \notag \\
		& \hspace{2cm} + \vert \varphi_t (x) \vert \; \|a_x^*a_x \psi_t \| \; \|b_y \psi_t \|  + \vert \varphi_t (x) \vert \; \|a_x  b_x \psi_t \|\; \| a_y \psi_t \|  \notag \\
		&\hspace{2cm}+ \vert \varphi_t (y) \vert \; \|a_x \psi_t \| \; \| a_y b_y^* \psi_t \| + \vert \varphi_t (y) \vert \; \|b_x \psi_t \| \; \| \psi_t \| + \vert \varphi_t (y) \vert \| b_x \psi_t \| \; \|a_y^* a_y \psi_t \| \bigg) .
	\end{align}
	To estimate the single terms of the r.h.s. we use that, for example the first term is bounded by 
	\begin{align}
		\|a_xb_x^* \psi_t \|^2 \leq \|n_x^{1/2} b_x^* \psi_t \|^2 = \langle \psi, b_x n_x b_x^* \psi \rangle  \; . 
	\end{align}
	With $n_xb_x^* = b_x^* (n_x +1 )$ resp $b_x n_x = (n_x +1 ) b_x$, 
	we furthermore get 
	\begin{align}
		\|a_xb_x^* \psi_t \|^2 \leq \langle \psi, b_x n_x b_x^* \psi \rangle = \langle \psi_t, (n_x +1)^{1/2} b_x b_x^* (n_x+1)^{1/2} \psi \rangle \leq \| (n_x + 1)^{1/2} \psi_t \|
	\end{align}
	and, we thus arrive at 
	\begin{align}
		&   \big\vert   \langle \psi_t, \big[ \mathcal{R}_{N,t}^{(2)}, a_x^*a_y ] \psi_t \rangle \big\vert \notag \\
		& \leq \frac{C}{\sqrt{N}} \bigg( \vert \varphi_t (x) \vert \; \|n_y^{1/2} \psi_t \| \; \|(n_x+1) \psi_t \| +  \vert \varphi_t (x) \vert \| n_y^{1/2} \psi_t \| \; \| \psi_t \| \notag \\
		& \hspace{2cm} + \vert \varphi_t (y) \vert \; \|n_x^{1/2} \psi_t \| \; \|(n_y +1 ) \psi_t \| +  \vert \varphi_t (y) \vert \| n_x^{1/2} \psi_t \| \; \| \psi_t \|\ \bigg) \; . \label{eq:estimate-R2}
	\end{align}
	
	It remains to control the commutator with the third remainder 
	\begin{align}
		\mathcal{R}_{N,t}^{(3)} = \frac{1}{2N}\sum_{w \in {\Zb^d}} a_w^*a_w^*a_wa_w  \; . 
	\end{align}
	We compute 
	\begin{align}
		\big[\mathcal{R}_{N,t}^{(3)}, a_x^*a_y\big] =& \frac{1}{N} \bigg( a_x^*a_x^*a_x a_y - a_x^*a_y^*a_y a_y \bigg)  \notag \\
		=& \frac{1}{N}  \bigg( (\mathcal{N}+1)^{-1/2} a_x^*a_x^*a_x a_y (\mathcal{N}+1)^{1/2}- (\mathcal{N}+1)^{1/2}a_x^*a_y^*a_y a_y(\mathcal{N}+1)^{-1/2}  \bigg)
	\end{align}
	where for the last equality we used that $\mathcal{N} a_x^*a_y^*a_ya_x  = a_x^*a_y^*a_ya_x \mathcal{N}$. With similar ideas as before, we get 
with the notation $\psi_t = e^{-iH_0 t} \psi$ for any $\psi \in \mathcal{F}_{\perp \varphi_t}^{\leq N}$, bounded  by
\begin{align}
\big\vert \langle \psi_t, \big[ \mathcal{R}_{N,t}^{(3)}, a_x^*a_y] \; \psi_t \rangle \big\vert \leq& \frac{1}{N}  \|a_x^*a_xa_x( \mathcal{N}+1)^{-1/2} \psi_t \|\; \| a_y ( \mathcal{N}+1)^{1/2} \psi_t \| \notag \\
&+ \frac{1}{N}  \|a_y^*a_ya_y( \mathcal{N}+1)^{-1/2} \psi_t \|\; \| a_x (\mathcal{N}+1)^{1/2} \psi_t \|  
\end{align}   
We recall that $a_x^*a_x = n_x $ that leads to 
\begin{align}
\|a_x^*a_xa_x ( \mathcal{N} + 1)^{-1/2} \psi_t \| = \|n_x a_x ( \mathcal{N} + 1)^{-1/2} \psi_t \| = \|a_x (n_x-1) ( \mathcal{N} + 1)^{-1/2} \psi_t \| ,
\end{align}
and since $n_x \leq \mathcal{N}$ we arrive at 
\begin{align}
\|a_x^*a_xa_x ( \mathcal{N} + 1)^{-1/2} \psi_t \| \leq  \|n_x^{1/2} (n_x-1) ( \mathcal{N} + 1)^{-1/2} \psi_t \| \leq \| n_x^{1/2} ( n_x+1)^{1/2} \psi_t \| 
\end{align}
so that we finally get 
\begin{align}
\big\vert \langle \psi_t, \big[ \mathcal{R}_{N,t}^{(3)}, a_x^*a_y] \; \psi_t \rangle \big\vert 
\leq& \frac{1}{N}  \|(n_x(n_x+1))^{1/2} \psi_t \|\; \| n_y^{1/2} ( \mathcal{N}+1)^{1/2} \psi_t \| \notag \\
&+ \frac{1}{N}  \|(n_y(n_y+1))^{1/2} \psi_t \|\; \| n_x^{1/2}(\mathcal{N}+1)^{1/2} \psi_t \| \; . \label{eq:estimate-R3}
\end{align}   
Summing up \eqref{eq:estimate-R0}, \eqref{eq:estimate-R1,1}, \eqref{eq:estimate-R1,2}, \eqref{eq:estimate-R2} and \eqref{eq:estimate-R3} and using that $\sum_{x,y \in {\Zb^d}} \big( g(x,y) + g(y,x) \big) = 2 \sum_{x,y \in {\Zb^d}} g(x,y)$, we get  the desired bound of Lemma \ref{lem42}.

\qed

\subsection{Proof of \lemref{lem44}}\label{secPfLem44}
Recall the definition 
$$\Phi(z,w):= e^{2\abs{w}-\frac{\g}{2}\abs{z-w}} \del{\sum_x  \abs{u_{x,z} (t)} }\del{\sum_y \abs{u_{y,z} (t)}\abs{\varphi_t(y)} }.$$
Owning to \eqref{uEst}, for $\g\ge1$, the $x$-sum is bounded as
\begin{align}
\label{418'}
\sum_x  \abs{u_{x,z} (t)} \le C_1e^{\abs{z}}\sum_x e^{-\g\abs{x-z}} \le C_de^{\abs{z}}. 
\end{align}

To bound the $y$-sum, let $\tilde \eps = 3\eps$ and 	$\tilde M$ to be determined later.
Consider the   decomposition of $\Zb^d$ as shown in  Fig.~\ref{figSplit}:

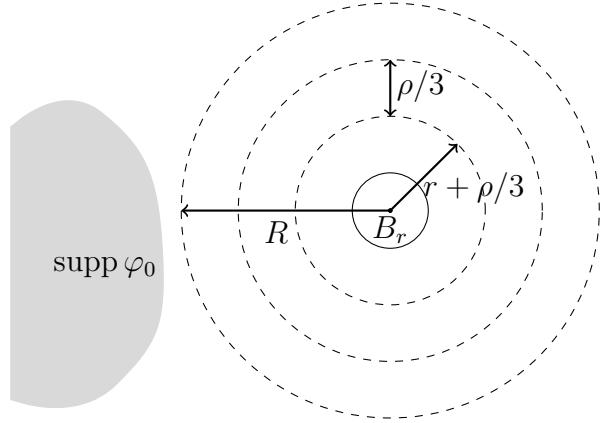
\begin{figure}[H]
\centering
\begin{tikzpicture}[scale=.5]
	\draw (0,0) circle (1) node at (0,-.5) {$B_r$};
	
	\filldraw (0,0) circle (0.05);
	
	\draw[dashed] (0,0) circle (5.5);
	
	\draw[dashed] (0,0) circle (2.5);
	
	\draw[dashed] (0,0) circle (4);



	\draw[thick, ->] (0,0) to (1.414*1.25,1.414*1.25);
	
	\draw[thick, ->] (0,0) to (-5.5,0);
	
	\node [below] at (-3,0) {$R$};
	
	\node at (2.2,.6 ) {$r+\rho/3$};
	
	\draw[thick, <->] (0,2.5) to (0,4);
	
	\node at (0.8,3.3) {$\rho/3$};
	
	
	
	
	\filldraw[gray!30] plot[shift={(0,-1)},scale = 2, smooth, tension=1] coordinates {(-5,1.6) (-3.8,1.8) (-3,0) (-3.6,-1.8) (-5,-2)};
	
	\node at (-7.5,- 1.5) {$\supp\varphi_0$ };
\end{tikzpicture}
\caption{Schematic diagram for the geometric splitting.} 
\label{figSplit}
\end{figure}
We distinguish two cases: 

\underline{Case 1.$\abs{z}> r+ \rho/3$.}

In this case, we use the fact  $\norm{\varphi_t}_{\ell^\infty}\le \norm{\varphi_t}_{\ell^2}\equiv 1$ to obtain
$$\sum_y \abs{u_{y,z} (t)}\abs{\varphi_t(y)} \le C_de^{\abs{z}} .$$
This, together with \eqref{PhiDef}, \eqref{418'}, yields 
\begin{align}
\label{423}
\Phi(z,w)\le& C_de^{2\abs{z}+2\abs{w}-\frac{\g}{2}\abs{z-w}}.
\end{align}
By assumption, we have $R\ge 2r$ and so $\rho\ge r$. Thus for every $\abs{w}\le r$ and $\abs{z}>r+\rho/3\ge 4r/3$, we have $\abs{z-w}\ge \frac14\abs{z}, \frac13\abs{w}.$ Plugging this back to \eqref{423} shows that
\begin{align}
\label{}
\Phi(z,w)\le C_de^{-\frac\g{12}\rho} ,\qquad \abs{z}> r+ \rho/3,\g\ge96.
\end{align}

\underline{Case 2. $\abs{z}\le r+ \rho/3$.}

In this case, we split the sum in \eqref{420} as
\begin{align}	
\sum_y \abs{u_{y,z} (t)}\abs{\varphi_t(y)} 	=&  {\sum_{\abs{y-z}\le \rho/3} \abs{u_{y,z} (t)}\abs{\varphi_t(y)}+\sum_{\abs{y-z}> \rho/3} \abs{u_{y,z} (t)}\abs{\varphi_t(y)}}.\label{425}
\end{align}
For the first sum in the r.h.s., note that $\abs{y}\le r+2\rho/3$. By the localization bound \eqref{NLSDL}, the assumption $\varphi_0(y)=0$ for $\abs{y}\ge R$, and the   bound \eqref{uEst}, there exists $\l_1=\l_1(\tilde M,\,\tilde\eps,\rho)>0$ s.th. for all $\l\ge \l_1$, 
\begin{align}
\label{426}
&\sum_{\abs{y-z}\le \rho/3} \abs{u_{y,z} (t)}\abs{\varphi_t(y)}
\notag\\	
\le&				\del{\sum_{\abs{y-z}\le \rho/3} \abs{u_{y,z} (t)}^2}^{1/2}\del{\sum_{\abs{y-z}\le \rho/3} \abs{\varphi_t(y)}^2}^{1/2}\notag\\
\le&   C_de^{\abs{z}}e^{-\frac{\tilde M}{2\tilde\eps}\rho},\quad  0\le t\le  {\rho}/{\tilde\eps}.
\end{align}
For the second sum in the r.h.s.~of \eqref{425}, 
recall that $\rho\ge r$ by the assumption $R\ge 2r$. Thus we have $\abs{z}\le \frac43 \rho$ and therefore,  by \eqref{uEst},  we have $\abs{u_{y,z}(t)}< C_1 e^{\frac{ 4-\g }{3}\rho }$ for $\abs{y-z}>\rho/3$,
Since, moreover, $\norm{\varphi_t}_{\ell^\infty}\le \norm{\varphi_t}_{\ell^2}\equiv1$, we find
\begin{align}
\label{427}
\sum_{\abs{y-z}> \rho/3} \abs{u_{y,z} (t)}\abs{\varphi_t(y)}\le \sum_{\abs{y-z}> \rho/3} \abs{u_{y,z} (t)}    < C_d e^{\abs{z}-\frac{\g }{3}\rho }.
\end{align}
Combining \eqref{426}--\eqref{427}, \eqref{PhiDef}, and \eqref{418'}, and using that $\abs{w}\le r\le \rho$ and $\abs{z}\le \frac43\rho$, we arrive at    
\begin{align}
\label{}
\Phi(z,w)\le& C_de^{2\abs{w} +2\abs{z}}(e^{-\frac{\tilde M}{2\tilde\eps}\rho}+e^{ -\frac{\g }{3}\rho })\notag\\
\le&C_d (e^{-(\frac{\tilde M}{2\tilde\eps}-\frac{14}3)\rho}+e^{-(\frac{\g }{3}-\frac{14}3)\rho }),  \qquad 0\le t\le \rho/\tilde\eps,\abs{z}\le r+ \rho/3.\notag
\end{align}
Choosing  $\tilde M = \frac{28}{3}\tilde\eps+2M$ yields
\begin{align}
\label{}
\Phi(z,w)\le C_d(e^{-\frac{M}{\tilde \eps}\rho}+e^{-\frac\g{12}\rho}), \qquad 0\le t\le \rho/\tilde\eps,\abs{z}\le r+ \rho/3,\g\ge \frac{56}3.
\end{align}			
This completes the bound for the $y$-sum in \eqref{PhiDef}.

Combining   cases 1--2 above, we conclude that for any $z\in\Zb^d$, 
\begin{align}
\label{318}
\Phi(z,w) \le C_{d} (e^{-\frac{M}{\tilde \eps}\rho} + e^{-\frac{\g}{12}\rho})\wndel{ (\cN+1) },\qquad 0\le t\le \rho/\tilde\eps, \g\ge96 . 
\end{align}
Choose now $\l_2\ge\l_1$ such that for $\l\ge\l_2$, $\g=\g(\l)$ in \eqref{1partDL} satisfies 
\begin{align}
\label{477}
\g \ge 96+12M/\tilde \eps.
\end{align}
Indeed, this is possible by assumption \ref{C1}, which requires 	 $\lim_{\l\to\infty}\g(\l)=\infty$.
Then it follows from \eqref{318} that $\Phi(z,w)\le  C_{d}e^{-\frac{M}{\tilde \eps}\rho}$, as desired. Note that \eqref{477} determines the threshold for $\l_*$ as a function of $\eps$ in the main results. 

\section{Proof of \thmref{thm2}}\label{secPfThm2}
At this stage, we can follow the line of arguments of \cite[Theorem 1 and Lemma 6.1]{LRZa} (see  \cite[Section 6]{LRZa}) and arrive at, with $\br{\cdot}_{t;s}$ defined in \eqref{def:Ast},
\begin{align}\label{115}
\big\vert &  \Tr \big( O ( \gamma_{\psi_{N,t}} - N \vert \varphi_t \rangle \langle \varphi_t  \vert ) \big) \big\vert \notag \\
&\leq C \|O \|_{\rm op} \bigg( \langle \mathcal{N}_{B_r}^+ \rangle_{(t;0)}+\| \varphi_t \|^2_{\ell^2( B_r)} \langle \mathcal{N}^+ \rangle_{(t;0)}\\
&\qquad\qquad\qquad + \| \varphi_t \|_{\ell^2( B_r)} \int_0^t ds \; e^{\vert U \vert \int_s^t dr \| \varphi_r \|_{\ell^\infty ( \mathbb{Z}^d)} ^2} \langle \mathcal{N}^+ + 1 \rangle_{(t;s)} \bigg)  \; . 
\end{align}
To estimate the first term of the r.h.s. of \eqref{115}, we use that by definition \eqref{rel1} and the slow propagation bound \eqref{slowPE} in Theorem \ref{thm1}. More precisely, by Theorem \ref{thm1} there exists $\lambda_0 = \lambda_0 (\eps, \rho, d)>0$, $C = C(r,d)>0 $ and $K = K(d, U)>0$ such that for all $\lambda \geq \lambda_0$ we have 
\begin{align}
\langle \mathcal{N}_{B_r}^+ \rangle_{(t;0)} \leq e^{K (t-\rho/\eps )} + \frac{e^{K t}}{N} \; .
\end{align}
Compare \eqref{slowPEsim}. 
For the second term of the r.h.s. of \eqref{115}, we find, since $\| \varphi_t \|_{\ell^\infty (B_r )} \leq \|\varphi_t \|_{\ell^2( B_r)}$, from the localization condition \eqref{NLSDL},  with $M=C(\abs{U} +K)$ for a generic constant $C>0$, that 
\begin{align}
\sup_{0\leq t \leq \rho/\eps} \| \varphi_t \|_{\ell^\infty (B_r)} \leq e^{-M \rho/\eps} \; 
\end{align}
so that we arrive with the global estimate on the number of excitations in Lemma \ref{lem23} at 
\begin{align}
\| \varphi_t \|^2_{\ell^2( B_r)} \langle \mathcal{N}^+ \rangle_{(t;0)} \leq e^{ C \vert U \vert t - M \rho / \eps} 
\end{align}
for a generic constant $C>0$. The remaining third term of the r.h.s. of  r.h.s. of \eqref{115} can be estimated similarly and we finally arrive at the desired estimate \eqref{114}.

\section{Proofs of Propositions \ref{Ex1} and \ref{Ex2}}\label{secPfEx1}

\begin{proof}[Proof of \propref{Ex1}]
We first verify the SUDL condition \ref{C1} with the quasi-periodic potential
$$v_{\theta,\al}(x)= \cos(2\pi (\theta+x\cdot \al)),\quad (\theta,\al)\in [0,1]\times[0,1]^d.$$  
Let $d\mu$ be a Borel probability measure on $[0,1]$, and let  $h_{\l,\theta,\al} = -\Lap + \l v_{\theta,\al}$. The following exponential dynamical localization result is proved in \cite{GYZa}: 
\begin{theorem}[\cite{GYZa}, Theorem 1.2]
If $\al\in[0,1]^d$ is Diophantine, then there exists $\l_0=\l_0(\al,d)>0$ s.th. for all $\l\ge\l_0$, there holds 
\begin{align}
	\label{EDL}
	\int_{[0,1]} \sup _{t \in \mathbb{R}}\left|\left\langle e^{-i t h_{\l,\theta,\al}}\delta_x,  \delta_y\right\rangle\right| d \mu(\theta) \leq C e^{-\tilde \gamma|x-y|}.
\end{align}
Furthermore, the exponent   $\tilde \g$ in the r.h.s.~satisfies $\tilde \g \ge \frac34\log {\l}$. 
\end{theorem}
Estimate \eqref{EDL}, together the simple argument in the proof of Theorem 7.6 in \cite{dRJLSa},  implies the SUDL condition \eqref{SUDL}  for a.e.~$\theta$ and $\g \ge \frac12 \log\l$. To be self-contain, we reproduce this argument here. 
Fix $\g<\tilde \g$ and consider 
$$
	Q(\theta)=\sum_{x,y}(1+|x|)^{-(d+1)} e^{\g|x-y|} \sup _t\left|\left\langle e^{-i t h_{\l,\theta,\al}}\delta_x,  \delta_y\right\rangle\right| .
$$
Plugging this to \eqref{EDL} shows that $\int_{[0,1]}Q(\theta)\,d\mu<\infty$, and so $Q(\theta)<\infty$ for a.e.~$\theta$. Therefore we have  $\sup _t\left|\left\langle e^{-i t h_{\l,\theta,\al}}\delta_x,  \delta_y\right\rangle\right| \leq C_\theta(1+|x|)^{d+1} e^{-\g|x-y|}$ for a.e.~$\theta$. For  $0<b\le1$, this gives the desired bound \eqref{SUDL} with $\g := \frac23 \tilde \g$.  This proves the claim, and \ref{C1} is verified.

Next, to verify \ref{C2}, we  recall the main result in \cite{CSW}. Consider the cubic NLS
\begin{align}
	\label{CSWeq}
	i\di_T\phi_T =v_{\theta,\al}\phi_T+a_1 \Lap\phi_T+a_2\abs{\phi_T}^2\phi_T.
\end{align}

\begin{theorem}[\cite{CSW}, Theorem 1.1]\label{thmCSZ}
	Let $\delta,\,\delta_*>0$ and $\rho\gg1$. Assume the initial state $\phi_0\in\ell^2(\Zb^d)$ satisfies
	\begin{align}
		\label{}
		\norm{\phi_0}_{\ell^2(B_r^\cp)}<\delta.
	\end{align}
	Then there exists $a_*=a_*(\delta_*,\rho,r,d)$ s.th. for $0<a:=a_1+a_2<a_*$,
	\begin{align}
		\label{andLoc}
		\norm{\phi_T}_{\ell^2(B_{r+\rho}^\cp)}<2\delta,\qquad  T\le \delta \cdot \eps^{-\rho^{1/2}},
	\end{align}
	on a set of $(\theta,\al)$ of measure at least $1-\delta_*$. 
\end{theorem}

We now shows how \thmref{thmCSZ} implies \ref{C2}, assuming $\supp \phi_0\subset B_r(x_0)$  with $\abs{x_0}>r+R$.
To this end, given $\l,\,U$ as in \eqref{NLS}, let 
$$a_1=\l^{-1}, \quad a_2=Ua_1,$$
and  consider the change of variable
\begin{align}
	\label{}
	\varphi_{t(T)}(x) :=\phi_T(x),\quad t(T)=\l^{-1} T.
\end{align}
Then direct computation shows that $\varphi_t$ solves \eqref{NLS} if and only if $\phi_T$ solves \eqref{CSWeq}. By \thmref{thmCSZ}, it follows that there exists $\l_1=\l_1 (\delta_*,\rho,r,d,U)>0$ s.th.~for all $\delta>0$ and  $\l>\l_1$, there holds
\begin{align}
	\label{122}
	\norm{\varphi_{t}}_{\ell^2([B_{r+\rho}(x_0)]^\cp)}<2\delta, 
\end{align}
in the time interval
\begin{align}
	\label{tauCond}
	0\le t\le  \delta   \l^{\rho^{1/2} {-1}}.
\end{align}
Now we show that \eqref{122} implies condition \ref{C2}. Given $\eps,\,M>0$,    choose
\begin{align}
	\label{}
	\delta = \frac12 e^{- M\rho/\eps}.
\end{align}
To ensure \ref{NLSDL}, we need to show that \eqref{122} holds with this choice of $\delta$ for all $t\le \rho/\eps$. In view of \eqref{tauCond}, we seek condition on $\l$ that gives
\begin{align}
	\label{514}
	e^{-M\rho/\eps} \l ^{\rho^{1/2}-1}\ge 2\rho/\eps .
\end{align}
Taking logarithm on both sides of \eqref{514} shows that a sufficient condition for \eqref{514} is $\l\ge \l_2$ with
\begin{align}
	\label{eCond}
	\l_2 =\max\left(\l_1,\exp (\frac{\log(\frac{2\rho}{\eps})+\frac{M\rho}{\eps}}{\sqrt{\rho} -1})\right).
\end{align}
Note that, as expected, for fixed $\eps$, the r.h.s.~of \eqref{eCond} tends to $\infty$ as $\rho\to \infty$.  Thus \ref{C2} is verified.

Finally, since the set of $\al$ satisfying the Diophantine condition has full measure, the set of pairs $(\theta,\al)$ that ensures \ref{C1}, while satisfying the assumption of \thmref{thmCSZ}, has measure at least $1-\delta_*$. 
Thus the proof of \propref{Ex1} is complete.
\end{proof}

\begin{remark}[On the size of $\eps$ for quasi-periodic models]\label{remSize}
From the proof above, we see that for given $\l$, we can choose $\eps$ to be of the order $\frac{1}{\log\l}$ in our main results, Thms.~\ref{thm2} and \ref{thm1}. 

Indeed, in   exactly two places in the proof of the main results are   constraints on $\l$ imposed:
\begin{itemize}
	\item In \eqref{477}, we require $\g=\g(\l)$ to be at least $96+12M/\tilde \eps$, where  $\tilde \eps = 3\eps$ and  $M$ is an $\cO(1)$ constant chosen later in \eqref{Mcond}, depending only on the interaction strength $U$ in \eqref{HNdef}. This, together with the lower bound $\g \ge \frac12 \log\l$ proved above, shows that a sufficient condition for \eqref{477} is $\eps \ge \frac{ c}{\log\l} $, for some constant $c>0$ depending only on $U$.  
	\item In \eqref{eCond}, we require $\l$ to be larger than $\l_1$ and $g_{\rho,M}(\eps):=\exp (\frac{\log(\frac{2\rho}{\eps})+\frac{M\rho}{\eps}}{\sqrt{\rho} -1})$.  Here, $\l_1$ is independent of $\eps$. For fixed $\rho$ and $M$, the function $\log g_{\rho,M}(\eps)=\cO(\eps^{-1})$ for $\eps\ll1$. Therefore, a sufficient condition for \eqref{eCond} is    $\eps \ge \frac{ \tilde c}{\log\l} $ for some $\tilde c>0$ depending only on $\delta_*,\rho,r,d$ and $U$.
\end{itemize}
\end{remark}
\begin{proof}[Proof of \propref{Ex2}]
To verify condition \ref{C1} in the random setting above, one combines Theorem 7.6 in \cite{dRJLSa} with Theorem 10.2 in \cite{aizenman2015random}.
To verify \ref{C2}, one uses the long-time Anderson localization for random NLS \cite{WZ2008}, which has been recently generalized to $\Zb^d,\,d\ge1$ in \cite{CSWa}. 
The rest of the proof is completely analogous to that of Prop.~\ref{Ex1} and we skip the details. 
\end{proof}

\section*{Acknowledgments}
The research of M.~L.~is supported by the DFG through the grant TRR 352 – Project-ID 470903074 and by
the European Union (ERC Starting Grant MathQuantProp, Grant Agreement 101163620)\footnote{Views and opinions expressed are however those of the authors only and do not necessarily reflect those of the European Union or the
European Research Council Executive Agency. Neither the European Union nor the granting authority can be held responsible for
them.}.  S.~R.~is supported by the European Research Council via the ERC CoG RAMBAS–Project–Nr.~10104424.
J.~Z.~is supported by National Key R \& D Program of China Grant 2022YFA100740, China Postdoctoral Science Foundation Grant 2024T170453, National Natural Science Foundation of China Grant 12401602, and the Shuimu Scholar program of Tsinghua University. 
He thanks L.~Ge, Y.~Sun, J.~You, and Q.~Zhou for helpful discussions.

\section*{Declarations}
\begin{itemize}
\item Conflict of interest: The authors have no financial or proprietary interests in any material discussed in this article.
\item Data availability: Data sharing is not applicable to this article as no datasets were generated or analyzed during the current study.
\end{itemize}

\bibliographystyle{alpha}
\bibliography{SlowFlucEst}

\end{document}